\documentclass[a4paper,10pt]{article}
\usepackage[english]{babel}
\usepackage[T1]{fontenc}
\usepackage[utf8]{inputenc}
\usepackage{amssymb}
\usepackage{physics}
\usepackage{amsmath}
\usepackage{graphicx}
\usepackage[unicode]{hyperref}
\usepackage{tikz}
\usepackage{color}
\usepackage{genyoungtabtikz}
\usepackage[title]{appendix}
\usepackage{indentfirst}
\usepackage{bm}
\usepackage{xhfill}
\usepackage{bbm}
\usepackage{multirow}
\usepackage{array}
\usepackage[centertableaux]{ytableau}
\usepackage{multicol}
\usepackage[shortlabels]{enumitem}
\usepackage{forest}
\usepackage{float}
\usepackage{mathrsfs}

\usetikzlibrary{decorations.pathreplacing}

\usepackage{ytableau}

\hypersetup{
    colorlinks=true,
    linkcolor=blue,
    filecolor=magenta,      
    urlcolor=cyan,
}
 
\usepackage{amsthm}

\theoremstyle{definition}

\usepackage{amsthm}

\theoremstyle{plain}
\newtheorem{proposition}{Proposition}[section]

\theoremstyle{definition}

\theoremstyle{remark}

\usepackage[nottoc]{tocbibind}

\usepackage{authblk}

\title{\textbf{Bulk Monodromy of Logarithmic Graviton Descendants in Critical Topologically Massive Gravity}}

\author[1,2]{\textbf{Yannick Mvondo-She}}
\affil[1]{National Institute of Theoretical and Computational Sciences, Private Bag X1, Matieland, South Africa}
\affil[1]{\texttt{yannick.mvondo-she@nithecs.ac.za}}
\affil[2]{Department of Physics, University of Pretoria, Private Bag X20, Hatfield 0028, \hspace{1cm} South Africa}

\date{}

\begin{document}

\maketitle

\begin{abstract}

We study the bulk analytic structure of the logarithmic graviton and
its global descendants in critical topologically massive gravity.
Starting from the Grumiller--Johansson mode, we complexify the
radial coordinate and derive its monodromy directly from the branch
structure and winding data of the logarithmic radial factor. We then
construct the global logarithmic descendants by explicit differential
action and show that each descendant decomposes into a universal
logarithmic contribution and a log-free meromorphic remainder. This
implies a universal unipotent monodromy throughout the global
descendant space. The associated nilpotent operator, obtained directly
from the bulk analytic continuation, is shown to intertwine the full
global $SL(2,\mathbb R)_L\times SL(2,\mathbb R)_R$ action. These results provide a direct bulk analytic realization of the logarithmic structure,
linking radial monodromy and global conformal symmetry.

\end{abstract}

\tableofcontents

\section{Introduction}
\label{sec:introduction}

Topologically massive gravity (TMG) provides one of the simplest settings in which local gravitational dynamics, asymptotic symmetry, and holography can be studied simultaneously in three spacetime dimensions. The theory is obtained by supplementing the Einstein--Hilbert action with negative cosmological constant by a gravitational Chern--Simons term \cite{Deser:1981wh,Deser:1982vy}. 

A revival of interest in three-dimensional topologically massive gravity (TMG) was catalyzed by a seminal work \cite{Li:2008dq}, which proposed that tuning the gravitational Chern-Simons coupling to a specific critical point yields a consistently unitary and stable theory termed \emph{chiral gravity}. On an $\mathrm{AdS}_3$ background of radius $\ell$, the linearized spectrum of the theory depends on the dimensionless combination $\mu\ell$, where $\mu$ is the topological mass parameter. In \cite{Li:2008dq}, a particular interest was attached to the chiral point $\mu\ell=1$, at which it was conjectured that the problematic negative-energy massive bulk ghosts decouple from the spectrum, leaving a theory holographically dual to a purely right-moving two-dimensional conformal field theory (CFT). The claim initiated an intense and rapid succession of investigations to scrutinise the stability and consistency of the vacuum, and subsequent analyses—most notably \cite{Carlip:2008jk,Giribet:2008bw}—contested that observation by pointing out the survival of non-trivial degrees of freedom. The persistence of physical, non-trivial logarithmic bulk modes at the critical point motivated the logarithmic relaxation of the theory, whose holographic description is naturally associated with a logarithmic CFT (LCFT). A definitive resolution of this conflict was provided in an immediate subsequent work \cite{Grumiller:2008qz}, which explicitly identified a physically propagating, negative-energy bulk mode characterized by logarithmic asymptotic growth. Rather than rendering the theory unphysical, this finding prompted the authors of \cite{Grumiller:2008qz} to broaden the asymptotic boundary conditions. This framework was later solidified through rigorous computations \cite{Grumiller:2008pr,Grumiller:2008es,Grumiller:2009mw,Ertl:2010dh,Gaberdiel:2010xv}, thereby establishing critical TMG as the conjectured gravity dual to a Logarithmic CFT (LCFT) \cite{Grumiller:2010rm,Grumiller:2010tj,Grumiller:2013at}.

The physical interpretation of Topologically Massive Gravity at the critical point depends entirely on the choice of asymptotic boundary conditions. Enforcing strict Brown-Henneaux asymptotically $\mathrm{AdS}_3$ boundary conditions \cite{Brown:1986nw} truncates the logarithmic sector to yield a unitary, purely right-moving chiral gravity. The chiral theory and its associated gravitational path integral have also motivated broader developments in three-dimensional quantum gravity \cite{Maloney:2015ina,Eberhardt:2022wlc}. The present paper does not study this Brown--Henneaux sector.

A different theory results when the boundary conditions are relaxed so
as to retain the additional solution that appears at the degeneracy
point. In this setting, critical TMG admits a logarithmic graviton whose asymptotic behavior violates the standard Brown--Henneaux falloff only logarithmically \cite{Grumiller:2008qz}. The corresponding theory,
also referred to as \textit{log gravity} \cite{Maloney:2009ck}, is the one considered in the present work. Its holographic interpretation is naturally associated with a logarithmic conformal field theory (LCFT), and holographic renormalization at the critical point reproduces the characteristic logarithmic two-point functions and new anomaly of the dual LCFT \cite{Skenderis:2009nt}. 

Logarithmic conformal field theories are characterized by
non-diagonalizable actions of the conformal generators and the
appearance of indecomposable representations
\cite{Gurarie:1993xq,Flohr:2001zs}. Their earlier relation to AdS/CFT \cite{Kogan:1999bn} and, in particular, to $\mathrm{AdS}_3/\mathrm{LCFT}_2$ holography \cite{Myung:1999nd} predates the application to critical TMG and provides the general representation-theoretic setting in which logarithmic bulk modes should be understood. Indecomposable and staggered Virasoro modules have been studied systematically in this broader LCFT framework \cite{Kytola:2009ax}, including their extension structure, descendant spaces, and logarithmic couplings. These results are important for the interpretation of critical TMG: the existence of a rank-two Jordan cell and the propagation of its indecomposable structure through a descendant representation are established features of logarithmic conformal theory.

Following the seminal derivation of the one-loop partition function of log gravity \cite{Gaberdiel:2010xv}, the combinatorics of the logarithmic sector was clarified by reformulating the logarithmic partition function via Bell polynomials \cite{Mvondo-She:2018htn}, subsequently establishing deep connections between critical TMG, integrable KP hierarchies, and Hurwitz numbers \cite{Mvondo-She:2019vbx,Mvondo-She:2021joh,Mvondo-She:2022jnf}.

For the logarithmic graviton of critical TMG, the ordinary left-moving
mode $\psi_{\mu\nu}^{L}$ and its logarithmic partner
$\psi_{\mu\nu}^{\log}$ form such a non-semisimple pair. In the
normalization used below, the total zero-mode operator
\begin{equation}
D=L_0+\bar L_0
\end{equation}
acts at the primary level according to
\begin{equation}
D\psi_{\mu\nu}^{L}
=
2\psi_{\mu\nu}^{L},
\qquad
D\psi_{\mu\nu}^{\log}
=
2\psi_{\mu\nu}^{\log}
+
\psi_{\mu\nu}^{L}.
\label{eq:intro_primary_jordan}
\end{equation}
The associated rank-two Jordan structure is therefore known from the
original construction of the logarithmic mode and its LCFT
interpretation. Similarly, once the Virasoro commutation relations are
given, the inheritance of this structure by descendants is an algebraic consequence.

The question addressed here is a bulk analytic one. The Grumiller--Johansson mode \cite{Grumiller:2008qz} contains the coordinate-dependent logarithmic factor
\begin{equation}
\psi_{\mu\nu}^{\log}
=
y\,\psi_{\mu\nu}^{L},
\qquad
y=-i\tau-\log\cosh\rho.
\label{eq:intro_exact_log_mode}
\end{equation}
It was recently observed that the defining algebraic structures of the logarithmic sector can be recovered entirely from the bulk geometry. Indeed, by complexifying the radial coordinate, it was shown \cite{Mvondo-She:2026egr}, that a unipotent monodromy action forces the linearized solutions into an indecomposable Jordan block representation. Crucially, this geometric monodromy is sufficiently constraining to reproduce the characteristic logarithmic two-point functions and mixing structures without a priori assumptions about the boundary LCFT data. After introducing a complex radial coordinate, the factor $y$ in Eq. (\ref{eq:intro_exact_log_mode}) becomes a genuinely multivalued function whose analytic continuation can be studied independently of the abstract Jordan representation. This raises a sequence of questions. What are the branch points of the radial logarithm? How does its continuation depend on the homotopy class of a contour in the complex radial plane? What happens to this analytic structure when the tensor field is acted upon repeatedly by the global lowering generators? Do the additional terms generated by differential action introduce further multivaluedness? Finally, if a monodromy operator is obtained from this bulk continuation, does it intertwine the complete global
$SL(2,\mathbb R)_L\times SL(2,\mathbb R)_R$ representation, including
the raising generators?

These questions distinguish the present construction from a purely
representation-theoretic argument. Starting from a known rank-two
Jordan pair and commuting its nilpotent map through the lowering
operators would immediately reproduce the same Jordan structure at
every descendant level. Such an argument, however, contains no
information about the analytic continuation of the coordinate-dependent
bulk fields. In particular, it does not determine the branch divisor
of the radial function, the winding-number dependence of its
monodromy, or the analytic character of the additional terms generated
by the differential action of the global conformal generators. The present work develops precisely this bulk construction. 

Writing the radial dependence of the logarithmic factor in terms of a complex coordinate $r$, we identify the meromorphic function
\begin{equation}
F(r)
=
\frac{r^2+1}{2r},
\label{eq:intro_F}
\end{equation}
and show that its zeros at $r=\pm i$ and its pole at $r=0$ determine
the logarithmic continuation. For a closed contour $C$ avoiding these
points, the resulting shift is
\begin{equation}
\Delta_C y
=
2\pi i
\left[
w_C(0)-w_C(i)-w_C(-i)
\right],
\label{eq:intro_winding_formula}
\end{equation}
where $w_C(p)$ denotes the winding number of $C$ about $p$. The
small-radius and asymptotic contour classes are consequently
inequivalent. For a positively oriented contour enclosing all three
distinguished points,
\begin{equation}
\Delta_\infty y=-2\pi i.
\label{eq:intro_asymptotic_shift}
\end{equation}
Thus the sign and normalization of the asymptotic monodromy follow from
the divisor and winding data of the  bulk radial function. 

A second ingredient is the explicit differential construction of the
global descendants. Defining
\begin{equation}
\psi_{\mu\nu;m,n}^{L}
=
L_{-1}^{\,m}\bar L_{-1}^{\,n}\psi_{\mu\nu}^{L},
\qquad
\psi_{\mu\nu;m,n}^{\log}
=
L_{-1}^{\,m}\bar L_{-1}^{\,n}\psi_{\mu\nu}^{\log},
\label{eq:intro_descendant_definitions}
\end{equation}
we show that every logarithmic descendant can be written as
\begin{equation}
\psi_{\mu\nu;m,n}^{\log}
=
y\,\psi_{\mu\nu;m,n}^{L}
+
\chi_{\mu\nu;m,n},
\label{eq:intro_general_decomposition}
\end{equation}
where the remainder $\chi_{\mu\nu;m,n}$ contains no logarithmic branch
and is meromorphic in the complex radial coordinate. The lowering
operators therefore generate nontrivial additional bulk terms, but
they do not generate additional logarithmic multivaluedness. The
unique multivalued contribution at every bi-grade remains
$y\psi_{\mu\nu;m,n}^{L}$.

This result promotes the primary radial continuation to an exact
statement across the global descendant family
\begin{equation}
M_\infty\psi_{\mu\nu;m,n}^{L}
=
\psi_{\mu\nu;m,n}^{L},
\qquad
M_\infty\psi_{\mu\nu;m,n}^{\log}
=
\psi_{\mu\nu;m,n}^{\log}
-
2\pi i\,\psi_{\mu\nu;m,n}^{L}.
\label{eq:intro_universal_monodromy}
\end{equation}
The uniform coefficient in
Eq.~\eqref{eq:intro_universal_monodromy} is thus obtained from the
analytic structure of the bulk descendants rather than assumed
from the known Jordan representation. This permits us to define a
bulk-derived nilpotent operator
\begin{equation}
\mathcal N_\infty
=
-\frac{1}{2\pi i}
\left(
M_\infty-\mathbf 1
\right),
\qquad
\mathcal N_\infty^2=0,
\label{eq:intro_Ninfty}
\end{equation}
whose action maps each logarithmic descendant to its ordinary partner.

A central test is then whether this independently constructed operator
is compatible with the global conformal representation. We compute the
actions of $L_0$, $\bar L_0$, $L_1$, and $\bar L_1$ explicitly. In
particular, the raising generators contain additional mixing terms on
the logarithmic tower, so their action is not obtained by simply
duplicating the ordinary highest-weight representation. Nevertheless,
we find
\begin{equation}
[\mathcal N_\infty,L_a]
=
[\mathcal N_\infty,\bar L_a]
=
0,
\qquad
a=-1,0,1,
\label{eq:intro_global_intertwining}
\end{equation}
and hence
\begin{equation}
[M_\infty,L_a]
=
[M_\infty,\bar L_a]
=
0.
\label{eq:intro_M_global_intertwining}
\end{equation}
The bulk monodromy therefore lies in the commutant of the global
$SL(2,\mathbb R)_L\times SL(2,\mathbb R)_R$ action on the logarithmic
descendant space. The nilpotent map familiar from the logarithmic
representation is thereby recovered from, and tested against, an
independently defined bulk analytic continuation. 

The remainder of the paper is organized as follows.
Section~\ref{sec:critical_tmg} reviews critical TMG, distinguishes
chiral gravity from log gravity, introduces the logarithmic
graviton, and fixes the global
$SL(2,\mathbb R)_L\times SL(2,\mathbb R)_R$ conventions.
Section~\ref{sec:radial_monodromy} analyzes the logarithmic
factor in the complex radial plane, derives its branch and
winding-number structure, and determines the asymptotic monodromy.
Section~\ref{sec:bulk_descendants} constructs the bulk
logarithmic descendants and establishes their decomposition into a
universal logarithmic term and a log-free meromorphic remainder.
Section~\ref{sec:descendant_monodromy} derives the universal monodromy
throughout the global descendant space and constructs the associated
nilpotent operator. Section~\ref{sec:global_action} computes the full
global conformal action and proves that the bulk-derived monodromy
intertwines $SL(2,\mathbb R)_L\times SL(2,\mathbb R)_R$. Section~\ref{sec:conclusions} summarizes our results and discusses extensions left for future work.

\section{Critical Topologically Massive Gravity and the Logarithmic Graviton}
\label{sec:critical_tmg}

Topologically massive gravity (TMG) provides a three-dimensional theory of
gravity in which the Einstein--Hilbert action with negative cosmological
constant is supplemented by a gravitational Chern--Simons term
\cite{Deser:1981wh,Deser:1982vy}. Around an $\mathrm{AdS}_3$ background,
the theory propagates a local massive spin-two degree of freedom whose
properties depend on the dimensionless combination $\mu\ell$, where
$\mu$ is the Chern--Simons coupling and $\ell$ is the AdS radius. The
critical value
\begin{equation}
\mu\ell=1
\label{eq:chiral_point}
\end{equation}
is distinguished by a degeneracy of the linearized equations: the massive
branch coalesces with the left-moving massless branch. The physical
interpretation of this degeneration depends crucially on the boundary
conditions imposed on the theory. With Brown--Henneaux boundary conditions,
the critical theory leads to the chiral-gravity construction
\cite{Brown:1986nw,Li:2008dq,Maloney:2009ck}; with logarithmically relaxed
boundary conditions, the same degeneracy admits an additional logarithmic
solution and gives rise to the log-gravity sector
\cite{Grumiller:2008qz,Grumiller:2008es,Skenderis:2009nt}. The present work
is concerned exclusively with the latter setting. In this section we fix
the bulk conventions used throughout the paper, review the distinction
between chiral gravity and log gravity, and introduce the Grumiller--Johansson logarithmic graviton together with the global
$SL(2,\mathbb R)_L\times SL(2,\mathbb R)_R$ generators that will be used
to construct its bulk descendants.

\subsection{Chiral gravity versus log gravity}

The special role of the point $\mu\ell=1$ in TMG was emphasized in \cite{Li:2008dq}. For asymptotically
$\mathrm{AdS}_3$ geometries satisfying the Brown--Henneaux boundary
conditions \cite{Brown:1986nw}, the left- and right-moving central charges
are
\begin{equation}
c_L
=
\frac{3\ell}{2G}
\left(
1-\frac{1}{\mu\ell}
\right),
\qquad
c_R
=
\frac{3\ell}{2G}
\left(
1+\frac{1}{\mu\ell}
\right),
\label{eq:TMG_central_charges}
\end{equation}
so that at the critical point
\begin{equation}
c_L=0,
\qquad
c_R=\frac{3\ell}{G}.
\label{eq:central_charges_chiral_point}
\end{equation}
Under Brown--Henneaux boundary conditions, the left-moving excitations
become pure gauge and the resulting theory is conjectured to define a
chiral quantum theory of gravity \cite{Li:2008dq}. The structure and
consistency of this chiral sector, including the role of the
$\mathrm{AdS}_3$ gravitational path integral and its spectrum, have been
studied extensively in subsequent work
\cite{Maloney:2009ck}.

The critical theory nevertheless admits a second possibility. Precisely
at $\mu\ell=1$, the degeneracy between the massive and left-moving
linearized operators produces an additional solution obtained by taking
the critical limit of the massive graviton, and it was shown that this mode has logarithmic asymptotic behaviour
\cite{Grumiller:2008qz}. It violates the standard Brown--Henneaux
fall-off conditions but is compatible with a consistent logarithmic
relaxation of them. The resulting theory is therefore not chiral gravity
with Brown--Henneaux boundary conditions, but the distinct theory usually
referred to as log gravity
\cite{Grumiller:2008es,Skenderis:2009nt}.

This distinction is essential for the present analysis. The boundary conditions that remove the logarithmic mode as pure gauge are not imposed here.
Instead, the logarithmic solution is retained as a physical element of the
linearized bulk solution space and we work with the logarithmically relaxed
asymptotic conditions appropriate to critical TMG. The resulting
non-diagonalizable structure is the gravitational realization of the
rank-two structure characteristic of logarithmic conformal field theory
(LCFT). Such indecomposable representations and their logarithmic
correlation functions were originally identified in two-dimensional
conformal field theory \cite{Gurarie:1993xq}, while their relation to
AdS/CFT and, in particular, to $\mathrm{AdS}_3/\mathrm{LCFT}_2$ was
developed in early holographic studies
\cite{Kogan:1999bn,Myung:1999nd}. In critical TMG, holographic
renormalization provides a direct realization of this correspondence:
the boundary two-point functions obtained from the bulk theory have the
characteristic LCFT form, with the left-moving stress tensor acquiring
a logarithmic partner and a corresponding new anomaly
\cite{Skenderis:2009nt}.

The starting point of the present paper will be the logarithmic tensor perturbation. This is important
for the analysis below because the logarithmic factor possesses a
nontrivial analytic structure in the complexified radial plane, and the
global conformal generators act as differential operators on the complete
bulk tensor field. Establishing these structures explicitly will allow us
to determine the monodromy of the gravitational descendants directly from
their bulk wavefunctions.

\subsection{Linearized equations at the chiral point}
\label{subsec:linearized_equations}

We now fix the bulk conventions and linearized equations used throughout
the remainder of the paper. Topologically massive gravity with negative
cosmological constant is described by the action
\begin{equation}
S_{\mathrm{TMG}}
=
\frac{1}{16\pi G}
\int d^3x\,\sqrt{-g}
\left(
R+\frac{2}{\ell^2}
\right)
+
\frac{1}{32\pi G\mu}
\int d^3x\,
\varepsilon^{\lambda\mu\nu}
\Gamma^{\rho}{}_{\lambda\sigma}
\left(
\partial_{\mu}\Gamma^{\sigma}{}_{\rho\nu}
+
\frac{2}{3}
\Gamma^{\sigma}{}_{\mu\tau}
\Gamma^{\tau}{}_{\nu\rho}
\right),
\label{eq:TMG_action}
\end{equation}
where \(G\) is Newton's constant, \(\ell\) is the AdS radius, and
\(\mu\) is the topological mass parameter
\cite{Deser:1981wh,Deser:1982vy}. Our convention for the cosmological
constant is
\begin{equation}
\Lambda=-\frac{1}{\ell^2}.
\label{eq:cosmological_constant}
\end{equation}

Variation of Eq.~\eqref{eq:TMG_action} gives the field equations
\begin{equation}
G_{\mu\nu}
-\frac{1}{\ell^2}g_{\mu\nu}
+
\frac{1}{\mu}C_{\mu\nu}
=
0,
\label{eq:TMG_field_equations}
\end{equation}
where
\begin{equation}
C_{\mu\nu}
=
\varepsilon_{\mu}{}^{\alpha\beta}
\nabla_{\alpha}
\left(
R_{\beta\nu}
-\frac{1}{4}g_{\beta\nu}R
\right)
\label{eq:Cotton_tensor}
\end{equation}
is the Cotton tensor. The maximally symmetric
\(\mathrm{AdS}_3\) background \(\bar g_{\mu\nu}\) satisfies
\begin{equation}
\bar R_{\mu\nu}
=
-\frac{2}{\ell^2}\bar g_{\mu\nu},
\qquad
\bar R
=
-\frac{6}{\ell^2},
\qquad
\bar C_{\mu\nu}=0.
\label{eq:AdS3_background_curvature}
\end{equation}

We perturb the metric according to
\begin{equation}
g_{\mu\nu}
=
\bar g_{\mu\nu}
+
h_{\mu\nu},
\label{eq:metric_perturbation}
\end{equation}
where \(h_{\mu\nu}\) is the linearized gravitational field. Throughout
the analysis we work in transverse--traceless gauge,
\begin{equation}
\bar\nabla^{\mu}h_{\mu\nu}=0,
\qquad
\bar g^{\mu\nu}h_{\mu\nu}=0.
\label{eq:TT_gauge}
\end{equation}
Thus the quantities denoted below by
\(\psi_{\mu\nu}^{L}\), \(\psi_{\mu\nu}^{M}\), and
\(\psi_{\mu\nu}^{\log}\) are explicit rank-two metric perturbations
satisfying Eq.~\eqref{eq:TT_gauge}, rather than scalar wavefunctions.

A particularly useful form of the linearized TMG equation is obtained
by introducing the mutually commuting first-order operators
\cite{Li:2008dq,Grumiller:2008qz}
\begin{equation}
\left(\mathcal D^{L/R}\right)_{\mu}{}^{\beta}
=
\delta_{\mu}^{\beta}
\pm
\ell\,
\varepsilon_{\mu}{}^{\alpha\beta}
\bar\nabla_{\alpha},
\label{eq:DL_DR_def}
\end{equation}
and
\begin{equation}
\left(\mathcal D^{M}\right)_{\mu}{}^{\beta}
=
\delta_{\mu}^{\beta}
+
\frac{1}{\mu}
\varepsilon_{\mu}{}^{\alpha\beta}
\bar\nabla_{\alpha}.
\label{eq:DM_def}
\end{equation}
In transverse--traceless gauge, the linearized equations can then be
written in factorized form as
\begin{equation}
\left(
\mathcal D^{L}
\mathcal D^{R}
\mathcal D^{M}
h
\right)_{\mu\nu}
=
0.
\label{eq:factorized_TMG}
\end{equation}
The three first-order factors describe, respectively, the left-moving,
right-moving, and massive branches of the linearized solution space.

For generic \(\mu\ell\), the massive solution is distinct from the two
massless branches. At the chiral point
\begin{equation}
\mu\ell=1,
\end{equation}
however, Eqs.~\eqref{eq:DL_DR_def} and \eqref{eq:DM_def} imply
\begin{equation}
\mathcal D^{M}
=
\mathcal D^{L}.
\label{eq:DM_equals_DL}
\end{equation}
The linearized field equation therefore degenerates to
\begin{equation}
\left(
\mathcal D^{R}
\left(\mathcal D^{L}\right)^2
h
\right)_{\mu\nu}
=
0.
\label{eq:critical_factorized_TMG}
\end{equation}
This repeated first-order factor is the bulk origin of the logarithmic
sector.

The ordinary left-moving graviton \(\psi_{\mu\nu}^{L}\) lies in the
kernel of \(\mathcal D^{L}\),
\begin{equation}
\mathcal D^{L}\psi^{L}=0.
\label{eq:left_mode_DL}
\end{equation}
At the critical point the degeneracy permits a generalized solution
\(\psi_{\mu\nu}^{\log}\) satisfying
\begin{equation}
\left(\mathcal D^{L}\right)^2
\psi^{\log}
=
0,
\qquad
\mathcal D^{L}\psi^{\log}\neq0.
\label{eq:log_generalized_kernel}
\end{equation}
With the Grumiller--Johansson normalization adopted below, the first
action of \(\mathcal D^{L}\) is proportional to the left-moving mode,
\begin{equation}
\mathcal D^{L}\psi^{\log}
\propto
\psi^{L}.
\label{eq:DL_log_proportional}
\end{equation}
We leave the proportionality factor implicit at this stage because the
subsequent monodromy analysis will be performed in the original
normalization of the logarithmic tensor mode, rather than by
rescaling it to an abstract canonical Jordan basis.

Equations~\eqref{eq:left_mode_DL}--\eqref{eq:DL_log_proportional}
exhibit the generalized-kernel structure responsible for the
non-diagonalizable logarithmic sector. The question addressed below is instead how the logarithmic tensor
solution behaves under the differential action of the global
$\mathrm{AdS}_3$ isometries and under analytic continuation in the
complexified radial coordinate. To formulate that problem, we next introduce the Grumiller--Johansson modes and
fix their normalization explicitly.

\subsection{The Grumiller--Johansson logarithmic mode}
\label{subsec:exact_GJ_mode}

We now introduce the bulk solutions that form the starting point
of the analysis. We use global coordinates
\((\tau,\phi,\rho)\), in which the $\mathrm{AdS}_3$ background is
\begin{equation}
ds^2
=
\ell^2
\left(
-\cosh^2\rho\,d\tau^2
+
\sinh^2\rho\,d\phi^2
+
d\rho^2
\right),
\label{eq:global_AdS3_metric}
\end{equation}
with
\begin{equation}
-\infty<\tau<\infty,
\qquad
0\leq\phi<2\pi,
\qquad
0\leq\rho<\infty.
\end{equation}
The conformal boundary is approached as
\(\rho\rightarrow\infty\). It is convenient to introduce the
light-cone coordinates
\begin{equation}
u=\tau+\phi,
\qquad
v=\tau-\phi.
\label{eq:lightcone_coordinates}
\end{equation}
Throughout the bulk analysis, $\tau$ and $\rho$ are initially taken to
be real. Complexification of the radial coordinate will be introduced
only when the analytic continuation and monodromy problem is considered
in Sec.~\ref{sec:radial_monodromy}.

The left-moving graviton at the chiral point is a highest-weight
solution of conformal weights
\begin{equation}
(h,\bar h)=(2,0).
\label{eq:left_weights}
\end{equation}
In the coordinate basis $(\tau,\phi,\rho)$, it may be written as \cite{Grumiller:2008qz}
\begin{equation}
\psi_{\mu\nu}^{L}
=
e^{-2iu}\tanh^2\rho\,
\begin{pmatrix}
1 & 1 & i b(\rho)
\\[1mm]
1 & 1 & i b(\rho)
\\[1mm]
i b(\rho) & i b(\rho) & -b^2(\rho)
\end{pmatrix}_{\mu\nu},
\label{eq:exact_left_graviton}
\end{equation}
where
\begin{equation}
b(\rho)
=
\frac{1}{\sinh\rho\,\cosh\rho}.
\label{eq:a_rho_definition}
\end{equation}
The overall normalization of
Eq.~\eqref{eq:exact_left_graviton} is conventional and will play no role
in the arguments below. What will be important is that
\(\psi_{\mu\nu}^{L}\) is an explicit rank-two tensor field satisfying
the transverse--traceless conditions
\eqref{eq:TT_gauge} and the left-moving first-order Eq. (\ref{eq:left_mode_DL}).

At the critical point, the degeneracy between the massive and
left-moving branches gives rise to the logarithmic mode found in \cite{Grumiller:2008qz}. In the normalization
used in their construction, the solution takes the particularly
simple multiplicative form
\begin{equation}
\psi_{\mu\nu}^{\log}
=
y(\tau,\rho)\,
\psi_{\mu\nu}^{L},
\label{eq:exact_log_graviton}
\end{equation}
where
\begin{equation}
y(\tau,\rho)
=
-i\tau-\ln\cosh\rho.
\label{eq:exact_y}
\end{equation}
Equation~\eqref{eq:exact_log_graviton} will be used throughout the present work.
In particular, the action of the global isometry generators below is
understood as the differential action on the complete tensor field
\(\psi_{\mu\nu}^{\log}\).

The origin of the logarithmic solution can be seen directly from the
critical degeneration of the massive graviton. Let
\(\psi_{\mu\nu}^{M}(\mu\ell)\) denote the massive branch away from the
critical point. As \(\mu\ell\rightarrow1\), this solution becomes
degenerate with the left-moving mode,
\begin{equation}
\psi_{\mu\nu}^{M}
\longrightarrow
\psi_{\mu\nu}^{L}.
\label{eq:massive_left_degeneracy}
\end{equation}
The logarithmic solution is obtained from the first nontrivial
variation of the massive mode at this degeneracy and therefore belongs
to the generalized kernel of the repeated critical operator,
\begin{equation}
\left(\mathcal D^L\right)^2
\psi^{\log}
=
0,
\qquad
\mathcal D^L\psi^{\log}\neq0.
\label{eq:exact_log_generalized_kernel}
\end{equation}
Thus the logarithmic graviton is an independent bulk solution of the
critical linearized equations rather than an additional ordinary
eigenmode of \(\mathcal D^L\).

The same generalized-eigenvector structure is visible under the
zero-mode generators. In the original normalization of
Eq.~\eqref{eq:exact_log_graviton}, one has
\begin{subequations}
\begin{align}
L_0\psi^L
&=
2\psi^L,
&
\bar L_0\psi^L
&=
0,
\label{eq:zero_modes_left}
\\[1mm]
L_0\psi^{\log}
&=
2\psi^{\log}
+
\frac{1}{2}\psi^L,
&
\bar L_0\psi^{\log}
&=
\frac{1}{2}\psi^L.
\label{eq:zero_modes_log_GJ}
\end{align}
\end{subequations}
It follows in particular that the total zero-mode operator
\begin{equation}
D:=L_0+\bar L_0
\label{eq:D_definition}
\end{equation}
acts as
\begin{subequations}
\begin{align}
D\psi^L
&=
2\psi^L,
\label{eq:D_left}
\\
D\psi^{\log}
&=
2\psi^{\log}
+
\psi^L.
\label{eq:D_log}
\end{align}
\end{subequations}
Thus, in the ordered basis
\begin{equation}
\Psi
=
\begin{pmatrix}
\psi^L\\[1mm]
\psi^{\log}
\end{pmatrix},
\end{equation}
the action of \(D\) is represented by
\begin{equation}
D
=
2\,\mathbf 1+N,
\qquad
N
=
\begin{pmatrix}
0&0\\
1&0
\end{pmatrix},
\qquad
N^2=0.
\label{eq:D_Jordan_GJ}
\end{equation}
Equivalently,
\begin{equation}
N\psi^{\log}=\psi^L,
\qquad
N\psi^L=0.
\label{eq:N_primary_action}
\end{equation}
We emphasize that Eqs.~\eqref{eq:zero_modes_left}--\eqref{eq:N_primary_action}
describe the established rank-two logarithmic structure of critical TMG;
this structure is not itself a new result of the present work. We retain
the normalization in Eq.~\eqref{eq:zero_modes_log_GJ} because it follows
directly from the bulk solution and will allow the normalization of
the radial monodromy to be tracked without introducing an intermediate
rescaling of the logarithmic state.

\subsection{Global $SL(2,\mathbb R)_L\times SL(2,\mathbb R)_R$ generators}
\label{subsec:global_generators}

The isometry group of global $\mathrm{AdS}_3$ is
\begin{equation}
SO(2,2)
\simeq
SL(2,\mathbb R)_L
\times
SL(2,\mathbb R)_R.
\label{eq:AdS3_isometry_group}
\end{equation}
In terms of the light-cone coordinates
\begin{equation}
u=\tau+\phi,
\qquad
v=\tau-\phi,
\end{equation}
introduced in Eq.~\eqref{eq:lightcone_coordinates}, a convenient basis
for the left-moving global generators is
\cite{Li:2008dq,Grumiller:2008qz}
\begin{equation}
L_0
=
i\partial_u,
\label{eq:L0_generator}
\end{equation}
together with
\begin{equation}
L_{-1}
=
i e^{-iu}
\left[
\frac{\cosh 2\rho}{\sinh 2\rho}\,\partial_u
-
\frac{1}{\sinh 2\rho}\,\partial_v
+
\frac{i}{2}\,\partial_\rho
\right],
\label{eq:Lminus1_generator}
\end{equation}
and
\begin{equation}
L_{1}
=
i e^{iu}
\left[
\frac{\cosh 2\rho}{\sinh 2\rho}\,\partial_u
-
\frac{1}{\sinh 2\rho}\,\partial_v
-
\frac{i}{2}\,\partial_\rho
\right].
\label{eq:Lplus1_generator}
\end{equation}
The right-moving generators are obtained by interchanging
\(u\leftrightarrow v\),
\begin{equation}
\bar L_0
=
i\partial_v,
\label{eq:Lbar0_generator}
\end{equation}
\begin{equation}
\bar L_{-1}
=
i e^{-iv}
\left[
\frac{\cosh 2\rho}{\sinh 2\rho}\,\partial_v
-
\frac{1}{\sinh 2\rho}\,\partial_u
+
\frac{i}{2}\,\partial_\rho
\right],
\label{eq:Lbarminus1_generator}
\end{equation}
and
\begin{equation}
\bar L_{1}
=
i e^{iv}
\left[
\frac{\cosh 2\rho}{\sinh 2\rho}\,\partial_v
-
\frac{1}{\sinh 2\rho}\,\partial_u
-
\frac{i}{2}\,\partial_\rho
\right].
\label{eq:Lbarplus1_generator}
\end{equation}
With these conventions, the two commuting global conformal algebras
satisfy
\begin{equation}
[L_0,L_{\pm1}]
=
\mp L_{\pm1},
\qquad
[L_1,L_{-1}]
=
2L_0,
\label{eq:left_sl2_algebra}
\end{equation}
and
\begin{equation}
[\bar L_0,\bar L_{\pm1}]
=
\mp\bar L_{\pm1},
\qquad
[\bar L_1,\bar L_{-1}]
=
2\bar L_0,
\label{eq:right_sl2_algebra}
\end{equation}
while
\begin{equation}
[L_a,\bar L_b]=0,
\qquad
a,b=-1,0,1.
\label{eq:left_right_commute}
\end{equation}

It is important to distinguish the coordinate expressions
\eqref{eq:L0_generator}--\eqref{eq:Lbarplus1_generator} from their
action on the gravitational perturbations. Let
\(\xi_a\) and \(\bar\xi_a\) denote the Killing vector fields represented
by \(L_a\) and \(\bar L_a\), respectively. Their action on a rank-two
covariant tensor perturbation is the corresponding Lie derivative,
\begin{equation}
L_a h_{\mu\nu}
\equiv
\mathcal L_{\xi_a}h_{\mu\nu},
\qquad
\bar L_a h_{\mu\nu}
\equiv
\mathcal L_{\bar\xi_a}h_{\mu\nu},
\label{eq:generators_as_Lie_derivatives}
\end{equation}
where
\begin{equation}
\left(\mathcal L_{\xi}h\right)_{\mu\nu}
=
\xi^\alpha\partial_\alpha h_{\mu\nu}
+
h_{\alpha\nu}\partial_\mu\xi^\alpha
+
h_{\mu\alpha}\partial_\nu\xi^\alpha.
\label{eq:tensor_Lie_derivative}
\end{equation}
Equivalently, using the background covariant derivative,
\begin{equation}
\left(\mathcal L_{\xi}h\right)_{\mu\nu}
=
\xi^\alpha\bar\nabla_\alpha h_{\mu\nu}
+
h_{\alpha\nu}\bar\nabla_\mu\xi^\alpha
+
h_{\mu\alpha}\bar\nabla_\nu\xi^\alpha.
\label{eq:covariant_tensor_Lie_derivative}
\end{equation}
Since \(\xi_a\) and \(\bar\xi_a\) are Killing vectors of the background,
their Lie derivatives map solutions of the linearized field equations
to solutions. The global conformal generators therefore act directly on
the bulk gravitational solution space.

For a scalar function \(f\), the same notation reduces to the ordinary
directional derivative,
\begin{equation}
L_a f
=
\xi_a^\alpha\partial_\alpha f,
\qquad
\bar L_a f
=
\bar\xi_a^\alpha\partial_\alpha f.
\label{eq:generator_action_scalar}
\end{equation}
Consequently, for a scalar \(f\) multiplying a tensor perturbation,
the Lie derivative obeys the Leibniz rule
\begin{equation}
L_a(fh_{\mu\nu})
=
(L_a f)h_{\mu\nu}
+
f\,L_a h_{\mu\nu},
\label{eq:Lie_derivative_Leibniz_left}
\end{equation}
and similarly
\begin{equation}
\bar L_a(fh_{\mu\nu})
=
(\bar L_a f)h_{\mu\nu}
+
f\,\bar L_a h_{\mu\nu}.
\label{eq:Lie_derivative_Leibniz_right}
\end{equation}
This elementary identity will become essential when the generators act
on the logarithmic graviton
\begin{equation}
\psi_{\mu\nu}^{\log}
=
y(\tau,\rho)\psi_{\mu\nu}^{L}.
\end{equation}
In particular,
\begin{equation}
L_a\psi_{\mu\nu}^{\log}
=
y\,L_a\psi_{\mu\nu}^{L}
+
(L_a y)\psi_{\mu\nu}^{L},
\label{eq:generator_log_decomposition_left}
\end{equation}
and
\begin{equation}
\bar L_a\psi_{\mu\nu}^{\log}
=
y\,\bar L_a\psi_{\mu\nu}^{L}
+
(\bar L_a y)\psi_{\mu\nu}^{L}.
\label{eq:generator_log_decomposition_right}
\end{equation}
Thus the differential action on the scalar logarithmic coefficient
cannot in general be omitted when constructing descendants of the
logarithmic tensor mode.

The left-moving mode is a global highest-weight state,
\begin{equation}
L_1\psi^L=0,
\qquad
\bar L_1\psi^L=0,
\label{eq:left_highest_weight_conditions}
\end{equation}
with weights given in Eq.~\eqref{eq:left_weights}. The logarithmic mode
satisfies the corresponding generalized highest-weight conditions
\begin{equation}
L_1\psi^{\log}=0,
\qquad
\bar L_1\psi^{\log}=0.
\label{eq:log_highest_weight_conditions}
\end{equation}
These relations can also be seen directly from the logarithmic
coefficient. Using
\begin{equation}
\partial_u y
=
\partial_v y
=
-\frac{i}{2},
\qquad
\partial_\rho y
=
-\tanh\rho,
\label{eq:y_basic_derivatives}
\end{equation}
one finds
\begin{equation}
L_1y
=
\frac{e^{iu}}{2}
\left(
\coth 2\rho
-
\csch 2\rho
-
\tanh\rho
\right)
=
0,
\label{eq:L1_y_zero}
\end{equation}
and likewise
\begin{equation}
\bar L_1y
=
\frac{e^{iv}}{2}
\left(
\coth 2\rho
-
\csch 2\rho
-
\tanh\rho
\right)
=
0,
\label{eq:Lbar1_y_zero}
\end{equation}
where we used
\begin{equation}
\coth 2\rho-\csch 2\rho=\tanh\rho.
\label{eq:hyperbolic_identity}
\end{equation}

The lowering generators behave differently. Their action on the same
scalar coefficient gives
\begin{equation}
L_{-1}y
=
\frac{e^{-iu}}{2}
\left(
\coth 2\rho
-
\csch 2\rho
+
\tanh\rho
\right)
=
e^{-iu}\tanh\rho,
\label{eq:Lminus1_y}
\end{equation}
and
\begin{equation}
\bar L_{-1}y
=
\frac{e^{-iv}}{2}
\left(
\coth 2\rho
-
\csch 2\rho
+
\tanh\rho
\right)
=
e^{-iv}\tanh\rho.
\label{eq:Lbarminus1_y}
\end{equation}
Therefore the lowering generators do not merely multiply the
left-moving descendant by the original logarithmic coefficient: they
also generate additional non-logarithmic bulk contributions.

This observation motivates the analysis developed in
Sec.~\ref{sec:bulk_descendants}. There we will apply
Eqs.~\eqref{eq:generator_log_decomposition_left} and
\eqref{eq:generator_log_decomposition_right} to the complete tensor
field and discuss the structure of the global descendants. Before constructing these descendants, however, we first
analyze the analytic structure of \(y(\tau,\rho)\) under
complexification of the radial coordinate. This is necessary in order
to define the relevant bulk monodromy independently of the abstract
global conformal representation.

\section{Complex Radial Structure and Asymptotic Bulk Monodromy}
\label{sec:radial_monodromy}

The logarithmic graviton introduced in
Sec.~\ref{sec:critical_tmg} contains all of its non-meromorphic radial
dependence through the scalar factor $y(\tau,\rho)$. When the radial dependence is analytically continued into the
complex domain, the function $\ln\cosh\rho$ determines the possible monodromies. The purpose of this section is therefore to
analyze the complex radial structure of the Grumiller--Johansson mode before assigning any monodromy to it. We will
show that the relevant continuation is controlled by a meromorphic
function with distinguished points at $r=0$ and $r=\pm i$, and that the
monodromy depends on the homotopy class of the chosen contour. In
particular, a small loop about the origin and a large asymptotic
$\mathrm{AdS}_3$ loop belong to inequivalent contour classes and induce
opposite logarithmic shifts.

\subsection{The logarithmic factor in the complex $r$-plane}
\label{subsec:log_factor_r_plane}

To expose the radial analytic structure, introduce the exponential
coordinate
\begin{equation}
r=e^\rho,
\qquad
\rho=\log r,
\label{eq:r_exp_rho}
\end{equation}
where $\log r$ denotes a local branch of the complex logarithm. On the
physical radial domain $\rho\geq0$, one has $r\geq1$, while analytic
continuation extends $r$ to the complex plane. Using
\begin{equation}
\cosh\rho
=
\frac{e^\rho+e^{-\rho}}{2},
\label{eq:cosh_rho_exp}
\end{equation}
we obtain
\begin{equation}
\cosh(\log r)
=
\frac{1}{2}
\left(
r+\frac{1}{r}
\right)
=
\frac{r^2+1}{2r}.
\label{eq:cosh_log_r}
\end{equation}
The logarithmic coefficient may therefore be written locally as
\begin{equation}
y(\tau,r)
=
-i\tau
-
\log\!\left(
\frac{r^2+1}{2r}
\right).
\label{eq:y_exact_r_plane}
\end{equation}
This form will be the starting point for the monodromy analysis.

We introduce the meromorphic function
\begin{equation}
F(r)
:=
\frac{r^2+1}{2r},
\label{eq:F_r_definition}
\end{equation}
so that
\begin{equation}
y(\tau,r)
=
-i\tau-\log F(r).
\label{eq:y_F_r}
\end{equation}
The analytic structure relevant to $y$ is consequently determined by
the divisor of $F$. The numerator factorizes as
\begin{equation}
r^2+1
=
(r-i)(r+i),
\label{eq:r2_plus_1_factorization}
\end{equation}
and hence $F$ has simple zeros at
\begin{equation}
r=i,
\qquad
r=-i,
\label{eq:F_zeros}
\end{equation}
and a simple pole at
\begin{equation}
r=0.
\label{eq:F_pole}
\end{equation}
Although $F(r)$ itself is meromorphic and therefore single-valued, the
composition $\log F(r)$ becomes multivalued upon analytic continuation
around its zeros and poles. Thus the finite complex radial plane relevant
to the logarithmic factor is naturally the punctured domain
\begin{equation}
X
=
\mathbb C\setminus\{0,i,-i\}.
\label{eq:punctured_r_plane}
\end{equation}
The points $r=0,\pm i$ are therefore the finite logarithmic branch
points of $y$.

An equivalent local decomposition follows from
Eq.~\eqref{eq:y_exact_r_plane}
\begin{equation}
y(\tau,r)
=
-i\tau
+
\ln2
+
\log r
-
\log(r^2+1),
\label{eq:y_local_log_decomposition}
\end{equation}
or, after factorization,
\begin{equation}
y(\tau,r)
=
-i\tau
+
\ln2
+
\log r
-
\log(r-i)
-
\log(r+i),
\label{eq:y_three_log_decomposition}
\end{equation}
where Eq.~\eqref{eq:y_three_log_decomposition} is understood locally,
with compatible choices of logarithmic branches. This expression makes
visible the relative signs with which the three distinguished points
enter the analytic continuation: the origin occurs with coefficient
$+1$, whereas $r=i$ and $r=-i$ occur with coefficient $-1$.

The same information can be extracted without introducing any explicit
branch cuts. Differentiating Eq.~\eqref{eq:y_F_r} with respect to $r$
gives
\begin{equation}
\frac{\partial y}{\partial r}
=
-\frac{F'(r)}{F(r)}
=
\frac{1}{r}
-
\frac{2r}{r^2+1}
=
\frac{1}{r}
-
\frac{1}{r-i}
-
\frac{1}{r+i}.
\label{eq:dy_dr_first_form}
\end{equation}
Thus the logarithmic differential
\begin{equation}
dy
=
\left[
\frac{1}{r}
-
\frac{1}{r-i}
-
\frac{1}{r+i}
\right]dr
\label{eq:dy_meromorphic_form}
\end{equation}
is a single-valued meromorphic one-form on $X$, with residues
\begin{equation}
\operatorname*{Res}_{r=0}dy=1,
\qquad
\operatorname*{Res}_{r=i}dy=-1,
\qquad
\operatorname*{Res}_{r=-i}dy=-1.
\label{eq:dy_residues}
\end{equation}
This formulation is preferable for the global continuation problem:
the multivaluedness of $y$ is encoded in the periods of the
single-valued one-form $dy$, and no arbitrary placement of branch cuts
is required.

The asymptotic region appears with clarity in the $r$-plane. For
$|r|>1$, Eq.~\eqref{eq:y_exact_r_plane} can be rearranged as
\begin{equation}
y(\tau,r)
=
-i\tau
-\log r
+\ln2
-\log(1+r^{-2}).
\label{eq:y_large_r_exact}
\end{equation}
Since
\begin{equation}
\log(1+r^{-2})
=
r^{-2}
-\frac{1}{2}r^{-4}
+\mathcal O(r^{-6}),
\qquad
|r|\rightarrow\infty,
\label{eq:log_one_plus_rminus2}
\end{equation}
one recovers
\begin{equation}
y(\tau,r)
=
-i\tau-\log r+\ln2
-r^{-2}
+\frac{1}{2}r^{-4}
+\mathcal O(r^{-6}).
\label{eq:y_large_r_expansion}
\end{equation}
The crucial point is that the
continuation in Eq.~\eqref{eq:y_exact_r_plane} contains the additional
finite singular points $r=\pm i$. Consequently, the monodromy of the
asymptotic term $-\log r$ cannot by itself be used to characterize an
arbitrary closed continuation in the full complex radial plane.

We therefore formulate the monodromy directly in terms of the periods
of Eq.~\eqref{eq:dy_meromorphic_form}. In the next subsection, the
residue theorem will give a contour-independent winding-number formula
for the logarithmic shift and will make precise how the three punctures
$r=0,\pm i$ contribute to the result.

\subsection{Branch points and winding-number formula}
\label{subsec:winding_number_formula}

Having identified the finite logarithmic branch points at
$r=0,\pm i$, we now determine the monodromy associated with an
arbitrary closed contour in the punctured radial plane
\begin{equation}
X=\mathbb C\setminus\{0,i,-i\}.
\end{equation}
The advantage of working with the meromorphic differential
\eqref{eq:dy_meromorphic_form} is that the logarithmic shift can be
computed without choosing explicit branch cuts for
$\log F(r)$.

Let $C\subset X$ be a closed, piecewise smooth, positively oriented
contour. Analytic continuation of $y$ along $C$ changes its value by
the period
\begin{equation}
\Delta_C y
:=
\oint_C dy.
\label{eq:monodromy_period_definition}
\end{equation}
Using Eq.~\eqref{eq:dy_meromorphic_form}, this becomes
\begin{equation}
\Delta_C y
=
\oint_C
\left[
\frac{1}{r}
-
\frac{1}{r-i}
-
\frac{1}{r+i}
\right]dr.
\label{eq:monodromy_period_integral}
\end{equation}
For a point $p\notin C$, define the winding number of $C$ about $a$ by
\begin{equation}
w_C(p)
=
\frac{1}{2\pi i}
\oint_C\frac{dr}{r-p}.
\label{eq:winding_number_definition}
\end{equation}
The three terms in Eq.~\eqref{eq:monodromy_period_integral} can
therefore be evaluated separately.

\begin{proposition}[Winding-number formula for the radial monodromy]
\label{prop:winding_monodromy}
Let
\begin{equation}
y(\tau,r)
=
-i\tau
-
\log\!\left(\frac{r^2+1}{2r}\right)
\end{equation}
be the logarithmic coefficient of the 
Grumiller--Johansson mode, and let
$C\subset\mathbb C\setminus\{0,i,-i\}$ be any closed contour.
Then analytic continuation along $C$ produces the additive monodromy
\begin{equation}
\Delta_C y
=
2\pi i
\left[
w_C(0)-w_C(i)-w_C(-i)
\right].
\label{eq:general_winding_monodromy}
\end{equation}
Hence the monodromy depends only on the winding numbers of $C$ around
the three finite punctures and is independent of the particular choice
of logarithmic branch cuts.
\end{proposition}

\begin{proof}
From Eq.~\eqref{eq:dy_meromorphic_form},
\begin{equation}
\Delta_C y
=
\oint_C\frac{dr}{r}
-
\oint_C\frac{dr}{r-i}
-
\oint_C\frac{dr}{r+i}
=
2\pi i\,w_C(0)
-
2\pi i\,w_C(i)
-
2\pi i\,w_C(-i),
\end{equation}
which immediately gives
Eq.~\eqref{eq:general_winding_monodromy}.
\end{proof}

The same result admits a useful interpretation in terms of the argument
principle. Since
\begin{equation}
y=-i\tau-\log F(r),
\qquad
F(r)=\frac{r^2+1}{2r},
\end{equation}
we have
\begin{equation}
dy=-d\log F(r).
\label{eq:dy_dlogF}
\end{equation}
For a closed contour $C$ on which $F$ has neither zeros nor poles,
\begin{equation}
\frac{1}{2\pi i}
\oint_C d\log F
=
N_Z(C)-N_P(C),
\label{eq:argument_principle_F}
\end{equation}
where $N_Z(C)$ and $N_P(C)$ denote, with winding multiplicities, the
numbers of zeros and poles of $F$ enclosed by the contour. It follows
that
\begin{equation}
\Delta_C y
=
-2\pi i
\left[
N_Z(C)-N_P(C)
\right].
\label{eq:argument_principle_monodromy}
\end{equation}
For the present function, the two simple zeros at $r=\pm i$ contribute
with the opposite sign to the simple pole at $r=0$. Thus the
monodromy is controlled by the divisor of the radial function
$F(r)$ rather than by the asymptotic term $\log r$ alone.

In particular, let $\gamma_0$, $\gamma_+$, and $\gamma_-$ denote
small positively oriented loops enclosing only $0$, $i$, and $-i$,
respectively. Proposition~\ref{prop:winding_monodromy} gives
\begin{equation}
\Delta_{\gamma_0}y=2\pi i,
\qquad
\Delta_{\gamma_+}y=-2\pi i,
\qquad
\Delta_{\gamma_-}y=-2\pi i.
\label{eq:elementary_monodromies}
\end{equation}
The three elementary continuations therefore do not induce identical
logarithmic shifts.

More generally, analytic continuation defines an additive monodromy
homomorphism
\begin{equation}
\mathfrak m:
\pi_1(X)
\longrightarrow
2\pi i\,\mathbb Z,
\label{eq:monodromy_homomorphism}
\end{equation}
whose action on a homotopy class $[C]$ is
\begin{equation}
\mathfrak m([C])
=
2\pi i
\left[
w_C(0)-w_C(i)-w_C(-i)
\right].
\label{eq:monodromy_homomorphism_explicit}
\end{equation}
The logarithmic continuation therefore factors through the
abelianization of the fundamental group: although
$\pi_1(X)$ retains the full information about the ordering of loops
around the three punctures, the additive logarithmic monodromy depends
only on their integer winding numbers.

An immediate consequence is that enclosing logarithmic branch points
does not necessarily imply nontrivial total monodromy. For example, a
positively oriented contour enclosing $r=0$ and $r=i$, but not
$r=-i$, has
\begin{equation}
w_C(0)=1,
\qquad
w_C(i)=1,
\qquad
w_C(-i)=0,
\end{equation}
and hence
\begin{equation}
\Delta_C y=0.
\label{eq:cancelling_monodromy_example}
\end{equation}
The contributions from the pole and zero of $F$ cancel exactly.
This cancellation illustrates why the multivaluedness of an isolated
$\log r$ term does not characterize the global analytic continuation
of the logarithmic graviton.

Equation~\eqref{eq:general_winding_monodromy} provides the general
monodromy law, but the contour relevant to the asymptotic
$\mathrm{AdS}_3$ region has not yet been selected. We address this
point next by comparing circles of radius $R<1$ and $R>1$. The two
classes are separated by the singular points $r=\pm i$ lying on the
unit circle and therefore cannot be continuously deformed into one
another within $X$.

\subsection{Small-radius and asymptotic contour classes}
\label{subsec:small_large_contours}

We now specialize the general winding-number formula
\eqref{eq:general_winding_monodromy} to circular contours centered at
the origin. For fixed radius \(R>0\), let
\begin{equation}
C_R:
\qquad
r(\theta)=Re^{i\theta},
\qquad
0\leq\theta\leq2\pi,
\label{eq:circular_contour}
\end{equation}
with counterclockwise orientation. Since the two finite zeros of
\(F(r)\) lie at
\begin{equation}
|i|=|-i|=1,
\end{equation}
the topology of \(C_R\) relative to the punctures changes precisely
when \(R\) crosses the unit circle. There are therefore two distinct
admissible regimes,
\begin{equation}
0<R<1,
\qquad\text{and}\qquad
R>1,
\label{eq:two_radial_regimes}
\end{equation}
while \(R=1\) is excluded because the contour passes through the
singular points \(r=\pm i\).

For \(0<R<1\), the contour \(C_R\) encloses the pole at the origin but
neither of the zeros at \(r=\pm i\). Hence
\begin{equation}
w_{C_R}(0)=1,
\qquad
w_{C_R}(i)=0,
\qquad
w_{C_R}(-i)=0.
\label{eq:small_circle_windings}
\end{equation}
Proposition~\ref{prop:winding_monodromy} therefore gives
\begin{equation}
\Delta_{C_R}y=2\pi i,
\qquad
0<R<1.
\label{eq:small_circle_monodromy}
\end{equation}
Thus a sufficiently small counterclockwise loop about \(r=0\) produces
the positive logarithmic shift
\begin{equation}
y\longmapsto y+2\pi i.
\label{eq:small_circle_y_shift}
\end{equation}

The situation is different for \(R>1\). A circle of radius greater
than unity encloses all three finite punctures,
\begin{equation}
r=0,\qquad r=i,\qquad r=-i,
\end{equation}
and therefore
\begin{equation}
w_{C_R}(0)
=
w_{C_R}(i)
=
w_{C_R}(-i)
=
1.
\label{eq:large_circle_windings}
\end{equation}
The general formula then yields
\begin{equation}
\Delta_{C_R}y
=
2\pi i(1-1-1) =
-2\pi i.
\end{equation}
Hence
\begin{equation}
\Delta_{C_R}y=-2\pi i,
\qquad
R>1,
\label{eq:large_circle_monodromy}
\end{equation}
or equivalently,
\begin{equation}
y\longmapsto y-2\pi i.
\label{eq:large_circle_y_shift}
\end{equation}

The opposite signs in
Eqs.~\eqref{eq:small_circle_monodromy} and
\eqref{eq:large_circle_monodromy} do not represent an ambiguity in the
choice of logarithmic branch. They correspond to genuinely different
homotopy classes in the punctured radial plane. Indeed, a continuous
deformation of \(C_R\) from \(R<1\) to \(R>1\) would necessarily pass
through \(R=1\), where
\begin{equation}
C_1\cap\{i,-i\}
=
\{i,-i\}.
\label{eq:unit_circle_intersection}
\end{equation}
Since the logarithmic coefficient is singular at these points, such a
deformation is not allowed within
\begin{equation}
X=\mathbb C\setminus\{0,i,-i\}.
\end{equation}
The small-radius and large-radius circles therefore cannot be
continuously deformed into one another while remaining in the domain of
analytic continuation.

The physically relevant contour for the asymptotic
\(\mathrm{AdS}_3\) analysis belongs to the second class. We therefore define the positively oriented asymptotic radial contour
\(C_\infty\) as any sufficiently large circle in this homotopy class.
Its logarithmic shift is
\begin{equation}
\Delta_{\infty}y=-2\pi i.
\label{eq:asymptotic_y_monodromy}
\end{equation}
Fig. \ref{fig:small_large_radial_contours} shows the two circular contour classes. 

\begin{figure}[htbp]
\centering

\begin{minipage}{0.47\textwidth}
\centering

\begin{tikzpicture}[scale=1.8]

\draw[->] (-1.7,0) -- (1.7,0)
    node[right] {$\operatorname{Re}(r)$};

\draw[->] (0,-1.7) -- (0,1.7)
    node[above] {$\operatorname{Im}(r)$};

\filldraw (0,0) circle (1.5pt)
    node[below right] {$r=0$};

\filldraw (0,1) circle (1.5pt)
    node[right] {$r=+i$};

\filldraw (0,-1) circle (1.5pt)
    node[right] {$r=-i$};

\draw[dashed] (0,0) circle (1);

\draw[thick] (0,0) circle (0.58);

\draw[->,thick]
    (0.58,0)
    arc[start angle=0,end angle=55,radius=0.58];

\node at (0.78,0.62) {$C_{\rm small}$};

\node[align=center] at (1,-1.45)
{
$0<R<1$\\[1mm]
$(w_0,w_i,w_{-i})=(1,0,0)$
};

\end{tikzpicture}

\vspace{1mm}

\[
\Delta_{C_{\rm small}}y
=
2\pi i
\]

\textbf{Local (small) contour}

\end{minipage}
\hfill
\begin{minipage}{0.47\textwidth}
\centering

\begin{tikzpicture}[scale=1.8]

\draw[->] (-1.7,0) -- (1.7,0)
    node[right] {$\operatorname{Re}(r)$};

\draw[->] (0,-1.7) -- (0,1.7)
    node[above] {$\operatorname{Im}(r)$};

\filldraw (0,0) circle (1.5pt)
    node[below right] {$r=0$};

\filldraw (0,1) circle (1.5pt)
    node[right] {$r=+i$};

\filldraw (0,-1) circle (1.5pt)
    node[right] {$r=-i$};

\draw[dashed] (0,0) circle (1);

\draw[thick] (0,0) circle (1.38);

\draw[->,thick]
    (1.38,0)
    arc[start angle=0,end angle=35,radius=1.38];

\node at (1.08,1.10) {$C_{\infty}$};

\node[align=center] at (1,-1.62)
{
$R>1$\\[1mm]
$(w_0,w_i,w_{-i})=(1,1,1)$
};

\end{tikzpicture}

\vspace{1mm}

\[
\Delta_{C_{\infty}}y
=
-2\pi i
\]

\textbf{Asymptotic (large) contour}

\end{minipage}

\caption{
Distinct homotopy classes of radial contours in the complex
$r$-plane for the logarithmic factor
$y(\tau,r)=-i\tau-\log[(r^2+1)/(2r)]$.
The small contour $C_{\rm small}$ encloses only the pole at
$r=0$, whereas the asymptotic contour $C_{\infty}$ encloses
$r=0$ and the two zeros $r=\pm i$.
Consequently,
$\Delta_{C_{\rm small}}y=+2\pi i$ while
$\Delta_{C_{\infty}}y=-2\pi i$.
The dashed circle indicates $|r|=1$, on which the points
$r=\pm i$ lie.
}
\label{fig:small_large_radial_contours}

\end{figure}

In the following subsection we apply this result to the complete
logarithmic tensor perturbation and determine the corresponding
asymptotic monodromy action on the Grumiller--Johansson pair.

\subsection{Asymptotic monodromy of the logarithmic primary}
\label{subsec:primary_asymptotic_monodromy}

Let
\(C_\infty\) denote a positively oriented large-radius contour in the
homotopy class \(R>1\), and let \(M_\infty\) denote analytic
continuation of a bulk solution along this contour. From
Eq.~\eqref{eq:asymptotic_y_monodromy}, the logarithmic
coefficient obeys
\begin{equation}
M_\infty y
=
y-2\pi i.
\label{eq:M_infty_y}
\end{equation}
To determine the monodromy of
\(\psi_{\mu\nu}^{\log}=y\psi_{\mu\nu}^{L}\), it remains to establish
the analytic behaviour of the left-moving tensor itself.

The radial dependence of the left-moving mode
\eqref{eq:exact_left_graviton} is built from
\(\tanh\rho\) and
\(b(\rho)=(\sinh\rho\cosh\rho)^{-1}\). Under
\begin{equation}
r=e^\rho,
\end{equation}
these functions become
\begin{equation}
\tanh\rho
=
\frac{r^2-1}{r^2+1},
\label{eq:tanh_rational_r}
\end{equation}
and
\begin{equation}
b(\rho)
=
\frac{1}{\sinh\rho\cosh\rho}
=
\frac{4r^2}{r^4-1}.
\label{eq:a_rational_r}
\end{equation}
Hence the radial coefficients appearing in
\(\psi_{\mu\nu}^{L}\) are rational functions of \(r\). More
explicitly,
\begin{equation}
\tanh^2\rho
=
\frac{(r^2-1)^2}{(r^2+1)^2},
\label{eq:tanh_squared_rational_r}
\end{equation}
so that every component of the tensor matrix in
Eq.~\eqref{eq:exact_left_graviton} is meromorphic in the complex
\(r\)-plane.

For the radial continuation considered here, the boundary coordinates
\(\tau\) and \(\phi\), and therefore \(u=\tau+\phi\), are held fixed.
The factor
\begin{equation}
e^{-2iu}
\end{equation}
is consequently unchanged along \(C_\infty\). Since a meromorphic
function is single-valued under analytic continuation around a closed
contour that avoids its poles, the left-moving tensor has no
logarithmic monodromy, i.e
\begin{equation}
M_\infty\psi_{\mu\nu}^{L}
=
\psi_{\mu\nu}^{L}.
\label{eq:left_primary_monodromy}
\end{equation}
Possible poles of the meromorphic continuation do not themselves
produce additive logarithmic shifts; the contour \(C_\infty\) is
understood not to pass through any singular point of the analytically
continued tensor field.

We may now continue the logarithmic solution
\begin{equation}
\psi_{\mu\nu}^{\log}
=
y\,\psi_{\mu\nu}^{L}.
\end{equation}
Using Eqs.~\eqref{eq:M_infty_y} and
\eqref{eq:left_primary_monodromy}, one obtains
\begin{equation}
M_\infty\psi_{\mu\nu}^{\log}
=
(M_\infty y)
(M_\infty\psi_{\mu\nu}^{L})
=
(y-2\pi i)\psi_{\mu\nu}^{L}
=
\psi_{\mu\nu}^{\log}
-
2\pi i\,\psi_{\mu\nu}^{L}.
\end{equation}
Therefore the asymptotic bulk monodromy of the
Grumiller--Johansson pair is
\begin{equation}
M_\infty\psi_{\mu\nu}^{L}
=
\psi_{\mu\nu}^{L},
\qquad
M_\infty\psi_{\mu\nu}^{\log}
=
\psi_{\mu\nu}^{\log}
-
2\pi i\,\psi_{\mu\nu}^{L}.
\label{eq:primary_monodromy_pair}
\end{equation}
The coefficient and sign in
Eq.~\eqref{eq:primary_monodromy_pair} are fixed by the radial
function and by the positively oriented asymptotic contour. In
particular, they are not introduced through a choice of canonical
Jordan normalization.

In the ordered basis
\begin{equation}
\Psi_{\mu\nu}
=
\begin{pmatrix}
\psi_{\mu\nu}^{L}
\\[1mm]
\psi_{\mu\nu}^{\log}
\end{pmatrix},
\label{eq:primary_pair_vector}
\end{equation}
the continuation is represented by the unipotent matrix
\begin{equation}
M_\infty
=
\begin{pmatrix}
1&0
\\[1mm]
-2\pi i&1
\end{pmatrix}.
\label{eq:primary_monodromy_matrix}
\end{equation}
Thus
\begin{equation}
(M_\infty-\mathbf 1)^2=0,
\qquad
M_\infty\neq\mathbf 1,
\label{eq:primary_unipotent_property}
\end{equation}
on the two-dimensional logarithmic solution space.

The relation between this unipotent continuation and the established
Jordan structure can be made explicit. Recall from
Eq.~\eqref{eq:N_primary_action} that the nilpotent map \(N\) acts as
\begin{equation}
N\psi^{\log}=\psi^L,
\qquad
N\psi^L=0,
\qquad
N^2=0.
\label{eq:N_primary_recalled}
\end{equation}
It follows immediately that
\begin{equation}
e^{-2\pi iN}
=
\mathbf 1-2\pi iN,
\label{eq:exp_nilpotent_primary}
\end{equation}
and hence
\begin{equation}
M_\infty
=
e^{-2\pi iN},
\label{eq:M_exp_N_primary}
\end{equation}
as an operator on the primary logarithmic pair.

The result \eqref{eq:primary_monodromy_pair} concerns the primary
logarithmic pair only. It does not by itself imply that the same
monodromy acts uniformly on the descendant space. Establishing that
statement requires the explicit bulk descendants, because the
lowering generators differentiate the logarithmic coefficient and
generate additional terms according to
Eqs.~\eqref{eq:generator_log_decomposition_left} and
\eqref{eq:generator_log_decomposition_right}. We therefore turn next
to the differential construction of the logarithmic descendants.

\section{Bulk Logarithmic Descendants}
\label{sec:bulk_descendants}

The asymptotic monodromy derived in
Sec.~\ref{sec:radial_monodromy} was obtained directly from the
Grumiller--Johansson primary pair. We now ask whether this bulk
analytic structure persists under the action of the global lowering
generators. This question cannot be answered merely by invoking the
Virasoro commutator with the zero modes: the generators
$L_{-1}$ and $\bar L_{-1}$ act as differential operators on the
complete tensor perturbation and therefore differentiate both the
left-moving tensor and its logarithmic coefficient. Starting from
Eq. (\ref{eq:intro_exact_log_mode}), we construct the descendants directly by Lie differentiation. We will
show that each logarithmic descendant separates naturally into a term
containing the original logarithmic factor $y$ multiplying the
corresponding left-moving descendant and an additional term generated
by differentiating $y$. The latter contains no logarithm. This
decomposition will provide the bulk input required in
Sec.~\ref{sec:descendant_monodromy} to determine the monodromy of the
entire global descendant space without assuming its form in advance.

\subsection{Differential action of $L_{-1}$}
\label{subsec:Lminus1_action}

We begin with the left-moving lowering generator. Define the first
left-moving descendants by
\begin{equation}
\psi_{\mu\nu;1,0}^{L}
:=
L_{-1}\psi_{\mu\nu}^{L},
\qquad
\psi_{\mu\nu;1,0}^{\log}
:=
L_{-1}\psi_{\mu\nu}^{\log},
\label{eq:first_left_descendants_definition}
\end{equation}
where, as specified in
Eq.~\eqref{eq:generators_as_Lie_derivatives}, the action of
$L_{-1}$ on a tensor field denotes the Lie derivative along the
corresponding $\mathrm{AdS}_3$ Killing vector.

Applying $L_{-1}$ to the logarithmic solution and using the Leibniz rule gives
\begin{equation}
\psi_{\mu\nu;1,0}^{\log}
=
L_{-1}\!\left(
y\,\psi_{\mu\nu}^{L}
\right)
=
y\,L_{-1}\psi_{\mu\nu}^{L}
+
(L_{-1}y)\psi_{\mu\nu}^{L}.
\label{eq:Lminus1_log_Leibniz}
\end{equation}
In terms of the left-moving descendant defined in
Eq.~\eqref{eq:first_left_descendants_definition},
\begin{equation}
\psi_{\mu\nu;1,0}^{\log}
=
y\,\psi_{\mu\nu;1,0}^{L}
+
(L_{-1}y)\psi_{\mu\nu}^{L}.
\label{eq:Lminus1_log_preliminary}
\end{equation}
The second term arises here because the lowering generator acts on the explicit coordinate dependence of the bulk logarithmic solution.

To evaluate this contribution, we use the fact that $\tau=(u+v)/2$ to rewrite the logarithmic coefficient as
\begin{equation}
y
=
-\frac{i}{2}(u+v)
-
\ln\cosh\rho,
\label{eq:y_uvrho}
\end{equation}  
and apply Eq.~\eqref{eq:Lminus1_generator} as
\begin{equation}
L_{-1}y
=
i e^{-iu}
\left[
-\frac{i}{2}\coth 2\rho
+
\frac{i}{2}\csch 2\rho
-
\frac{i}{2}\tanh\rho
\right]
=
\frac{e^{-iu}}{2}
\left[
\coth 2\rho
-
\csch 2\rho
+
\tanh\rho
\right],
\label{eq:Lminus1_y_intermediate}
\end{equation}
obtaining
\begin{equation}
L_{-1}y
=
e^{-iu}\tanh\rho.
\label{eq:Lminus1_y_section4}
\end{equation}
The first left-moving logarithmic descendant consequently takes the form
\begin{equation}
\psi_{\mu\nu;1,0}^{\log}
=
y\,\psi_{\mu\nu;1,0}^{L}
+
e^{-iu}\tanh\rho\,
\psi_{\mu\nu}^{L}
=
y\,\psi_{\mu\nu;1,0}^{L}
+
\chi_{\mu\nu;1,0},
\label{eq:first_left_log_descendant}
\end{equation}
as we isolate the additional term by defining
\begin{equation}
\chi_{\mu\nu;1,0}
:=
e^{-iu}\tanh\rho\,
\psi_{\mu\nu}^{L}.
\label{eq:chi_10_definition}
\end{equation}
Equation~\eqref{eq:first_left_log_descendant}  separates the two analytically distinct pieces of
the descendant: all explicit logarithmic dependence remains contained
in the first term, whereas $\chi_{\mu\nu;1,0}$ is generated entirely
by the differential action of $L_{-1}$ on $y$. From Eq.~\eqref{eq:tanh_rational_r}, it appears that $L_{-1}y$ is meromorphic in $r$ and contains no logarithmic
branch. Since $\psi_{\mu\nu}^{L}$ is itself meromorphic under the
radial complexification, as established in
Sec.~\ref{subsec:primary_asymptotic_monodromy}, it follows that
$\chi_{\mu\nu;1,0}$ is meromorphic in $r$. In particular, the additional term produced by the differential action does not introduce a new additive logarithmic monodromy, and is essential for tracking the analytic structure of the complete bulk tensor descendant.

The right-moving lowering generator produces an analogous, but
independent, contribution. We evaluate that action next before
combining the two directions into the general bi-graded descendant
family.

\subsection{Differential action of $\bar L_{-1}$}
\label{subsec:Lbarminus1_action}

We now repeat the construction in the right-moving lowering direction.
Define
\begin{equation}
\psi_{\mu\nu;0,1}^{L}
:=
\bar L_{-1}\psi_{\mu\nu}^{L},
\qquad
\psi_{\mu\nu;0,1}^{\log}
:=
\bar L_{-1}\psi_{\mu\nu}^{\log}.
\label{eq:first_right_descendants_definition}
\end{equation}
Although the primary left-moving mode has conformal weights
$(h,\bar h)=(2,0)$, the action of $\bar L_{-1}$ on the bulk tensor
field need not vanish. The global descendant space considered below
is therefore naturally bi-graded by the numbers of $L_{-1}$ and
$\bar L_{-1}$ actions.

Using again the logarithmic solution, the tensorial Leibniz rule gives
\begin{equation}
\psi_{\mu\nu;0,1}^{\log}
=
\bar L_{-1}
\left(
y\,\psi_{\mu\nu}^{L}
\right)
=
y\,\bar L_{-1}\psi_{\mu\nu}^{L}
+
(\bar L_{-1}y)\psi_{\mu\nu}^{L}
=
y\,\psi_{\mu\nu;0,1}^{L}
+
(\bar L_{-1}y)\psi_{\mu\nu}^{L},
\label{eq:Lbarminus1_log_Leibniz}
\end{equation}
and the first right-moving logarithmic descendant is therefore
\begin{equation}
\psi_{\mu\nu;0,1}^{\log}
=
y\,\psi_{\mu\nu;0,1}^{L}
+
e^{-iv}\tanh\rho\,
\psi_{\mu\nu}^{L}
=
y\,\psi_{\mu\nu;0,1}^{L}
+
\chi_{\mu\nu;0,1},
\label{eq:first_right_log_descendant}
\end{equation}
if in direct analogy with the left-moving case, we define
\begin{equation}
\chi_{\mu\nu;0,1}
:=
e^{-iv}\tanh\rho\,
\psi_{\mu\nu}^{L}.
\label{eq:chi_01_definition}
\end{equation}
As for $\chi_{\mu\nu;1,0}$, the additional term contains no explicit
logarithmic factor, and since the left-moving tensor
$\psi_{\mu\nu}^{L}$ is also meromorphic in the complex radial
coordinate, it follows that $\chi_{\mu\nu;0,1}$ is meromorphic in $r$.

The differential actions of $L_{-1}$ and $\bar L_{-1}$ thus have the
same structural effect: each preserves a single copy of the original
logarithmic coefficient multiplying the corresponding ordinary
descendant and generates an additional log-free term.


\subsection{General descendant decomposition}
\label{subsec:general_descendant_decomposition}

We now extend the level-one structure established in the preceding two subsections to arbitrary global descendant 
level. We adopt the bi-graded notation
\begin{equation}
\psi_{\mu\nu;m,n}^{L}
:=
L_{-1}^{\,m}\bar L_{-1}^{\,n}\psi_{\mu\nu}^{L},
\qquad
\psi_{\mu\nu;m,n}^{\log}
:=
L_{-1}^{\,m}\bar L_{-1}^{\,n}\psi_{\mu\nu}^{\log},
\label{eq:general_descendants_recalled}
\end{equation}
with
\begin{equation}
m,n\in\mathbb Z_{\geq0}.
\end{equation}
The primary logarithmic field satisfies
\begin{equation}
\psi_{\mu\nu;0,0}^{\log}
=
y\,\psi_{\mu\nu;0,0}^{L},
\end{equation}
so that it is natural to set
\begin{equation}
\chi_{\mu\nu;0,0}=0.
\label{eq:chi00_general}
\end{equation}

\begin{proposition}[Decomposition of global logarithmic descendants]
\label{prop:general_descendant_decomposition}
For every pair
\begin{equation}
(m,n)\in\mathbb Z_{\geq0}^2,
\end{equation}
the logarithmic bulk descendant admits the decomposition
\begin{equation}
\psi_{\mu\nu;m,n}^{\log}
=
y\,\psi_{\mu\nu;m,n}^{L}
+
\chi_{\mu\nu;m,n},
\label{eq:general_descendant_decomposition}
\end{equation}
where the remainder $\chi_{\mu\nu;m,n}$ satisfies the recursion relations
\begin{equation}
\chi_{\mu\nu;m+1,n}
=
L_{-1}\chi_{\mu\nu;m,n}
+
(L_{-1}y)\,
\psi_{\mu\nu;m,n}^{L},
\label{eq:chi_left_recursion}
\end{equation}
and
\begin{equation}
\chi_{\mu\nu;m,n+1}
=
\bar L_{-1}\chi_{\mu\nu;m,n}
+
(\bar L_{-1}y)\,
\psi_{\mu\nu;m,n}^{L}.
\label{eq:chi_right_recursion}
\end{equation}
\end{proposition}

\begin{proof}
The statement holds at $(m,n)=(0,0)$ because
\begin{equation}
\psi_{\mu\nu}^{\log}
=
y\,\psi_{\mu\nu}^{L},
\qquad
\chi_{\mu\nu;0,0}=0.
\end{equation}
Assume now that for some $(m,n)$ one has
\begin{equation}
\psi_{\mu\nu;m,n}^{\log}
=
y\,\psi_{\mu\nu;m,n}^{L}
+
\chi_{\mu\nu;m,n}.
\label{eq:induction_hypothesis}
\end{equation}
Acting with $L_{-1}$ and using the tensorial Leibniz rule gives
\begin{subequations}
\begin{align}
\psi_{\mu\nu;m+1,n}^{\log}
&=
L_{-1}\psi_{\mu\nu;m,n}^{\log}
\\[1mm]
&=
L_{-1}
\left(
y\,\psi_{\mu\nu;m,n}^{L}
+
\chi_{\mu\nu;m,n}
\right)
\\[1mm]
&=
y\,L_{-1}\psi_{\mu\nu;m,n}^{L}
+
(L_{-1}y)\psi_{\mu\nu;m,n}^{L}
+
L_{-1}\chi_{\mu\nu;m,n}.
\end{align}
\end{subequations}
Since
\begin{equation}
L_{-1}\psi_{\mu\nu;m,n}^{L}
=
\psi_{\mu\nu;m+1,n}^{L},
\end{equation}
this becomes
\begin{equation}
\psi_{\mu\nu;m+1,n}^{\log}
=
y\,\psi_{\mu\nu;m+1,n}^{L}
+
\chi_{\mu\nu;m+1,n},
\end{equation}
with $\chi_{\mu\nu;m+1,n}$ given by
Eq.~\eqref{eq:chi_left_recursion}.

The same argument with $\bar L_{-1}$ yields
\begin{equation}
\psi_{\mu\nu;m,n+1}^{\log}
=
y\,\psi_{\mu\nu;m,n+1}^{L}
+
\chi_{\mu\nu;m,n+1},
\end{equation}
with $\chi_{\mu\nu;m,n+1}$ given by
Eq.~\eqref{eq:chi_right_recursion}. Since the two lowering generators
commute,
\begin{equation}
[L_{-1},\bar L_{-1}]=0,
\end{equation}
the decomposition is independent of the order in which the two
descendant directions are generated. The result therefore holds for
all $(m,n)\in\mathbb Z_{\geq0}^2$.
\end{proof}

We now show that the remainder is meromorphic in the complex radial coordinate $r$ at every global descendant level and therefore carries no logarithmic branch monodromy.

\begin{proposition}[Meromorphicity of the descendant remainder]
\label{prop:remainder_meromorphic}
For every
\begin{equation}
(m,n)\in\mathbb Z_{\geq0}^2,
\end{equation}
the remainder
$\chi_{\mu\nu;m,n}$ is meromorphic in
the complex radial coordinate $r$. In particular, it contains no
logarithmic branch and acquires no additive logarithmic monodromy under
analytic continuation along a closed contour that avoids its poles.
\end{proposition}

\begin{proof}
We proceed recursively from
\begin{equation}
\chi_{\mu\nu;0,0}=0,
\end{equation}
which is trivially meromorphic.

Assume that $\chi_{\mu\nu;m,n}$ and
$\psi_{\mu\nu;m,n}^{L}$ are meromorphic in $r$. The left recursion
\eqref{eq:chi_left_recursion} gives
\begin{equation}
\chi_{\mu\nu;m+1,n}
=
L_{-1}\chi_{\mu\nu;m,n}
+
(L_{-1}y)\psi_{\mu\nu;m,n}^{L}.
\label{eq:meromorphic_induction_left}
\end{equation}
The first term on the right-hand side is meromorphic because
$L_{-1}$ is a first-order differential operator with meromorphic
coefficients. The second term is a product of the meromorphic function
$L_{-1}y$ and the meromorphic tensor
$\psi_{\mu\nu;m,n}^{L}$. Hence
$\chi_{\mu\nu;m+1,n}$ is meromorphic.

Similarly, from
\begin{equation}
\chi_{\mu\nu;m,n+1}
=
\bar L_{-1}\chi_{\mu\nu;m,n}
+
(\bar L_{-1}y)\psi_{\mu\nu;m,n}^{L},
\label{eq:meromorphic_induction_right}
\end{equation}
both terms on the right-hand side are meromorphic, and therefore
$\chi_{\mu\nu;m,n+1}$ is meromorphic.

The ordinary descendants are generated from the meromorphic primary
$\psi_{\mu\nu}^{L}$ by repeated action of the same meromorphic
differential operators,
\begin{equation}
\psi_{\mu\nu;m,n}^{L}
=
L_{-1}^{\,m}\bar L_{-1}^{\,n}
\psi_{\mu\nu}^{L},
\end{equation}
and are consequently meromorphic for all $(m,n)$. The recursion
therefore closes, proving the statement at every descendant level.
\end{proof}

Proposition~\ref{prop:remainder_meromorphic} makes precise the sense in
which the remainder is log-free. This yields a useful structural decomposition of the entire global descendant family
\begin{equation}
\psi_{\mu\nu;m,n}^{\log}
=
\underbrace{
y\,\psi_{\mu\nu;m,n}^{L}
}_{\text{multivalued logarithmic part}}
+
\underbrace{
\chi_{\mu\nu;m,n}
}_{\text{single-valued meromorphic part}},
\label{eq:analytic_descendant_decomposition}
\end{equation}
with all additive logarithmic multivaluedness of the descendant
concentrated in the single factor $y$.

 In the next section we use these results to derive the asymptotic monodromy at
arbitrary $(m,n)$ and identify the corresponding nilpotent operator on
the global descendant space.

\section{Monodromy Across the Global Descendant Space}
\label{sec:descendant_monodromy}

The preceding section established the bulk decomposition
given by Eq. (\ref{eq:general_descendant_decomposition}) in Proposition~\ref{prop:general_descendant_decomposition}. We now combine this bulk decomposition with the asymptotic
continuation derived in Sec.~\ref{sec:radial_monodromy}. The purpose is
to determine the monodromy of the complete global descendant space
directly from the analytic structure of its tensor wavefunctions,
rather than infer it from the Virasoro commutators. We first derive
the universal descendant monodromy, then identify the nilpotent operator
that generates its unipotent action.

\subsection{Universal asymptotic monodromy}
\label{subsec:universal_asymptotic_monodromy}

We are now in a position to determine the asymptotic monodromy of an
arbitrary global logarithmic descendant. 
\begin{proposition}[Universal asymptotic monodromy]
\label{prop:universal_asymptotic_monodromy}
Let
\begin{equation}
\psi_{\mu\nu;m,n}^{L}
=
L_{-1}^{\,m}\bar L_{-1}^{\,n}\psi_{\mu\nu}^{L},
\qquad
\psi_{\mu\nu;m,n}^{\log}
=
L_{-1}^{\,m}\bar L_{-1}^{\,n}\psi_{\mu\nu}^{\log},
\end{equation}
with $m,n\in\mathbb Z_{\geq0}$. Under analytic continuation along a
positively oriented asymptotic radial contour $C_\infty$, the complete
global descendant family transforms according to
\begin{equation}
M_\infty\psi_{\mu\nu;m,n}^{L}
=
\psi_{\mu\nu;m,n}^{L},
\label{eq:universal_left_monodromy}
\end{equation}
and
\begin{equation}
M_\infty\psi_{\mu\nu;m,n}^{\log}
=
\psi_{\mu\nu;m,n}^{\log}
-
2\pi i\,\psi_{\mu\nu;m,n}^{L}.
\label{eq:universal_log_monodromy}
\end{equation}
The mixing coefficient $-2\pi i$ is independent of the descendant
labels $(m,n)$.
\end{proposition}

\begin{proof}
From Proposition~\ref{prop:general_descendant_decomposition}, the
logarithmic descendant has the form
\begin{equation}
\psi_{\mu\nu;m,n}^{\log}
=
y\,\psi_{\mu\nu;m,n}^{L}
+
\chi_{\mu\nu;m,n}.
\end{equation}
Proposition~\ref{prop:remainder_meromorphic} implies that the remainder
is meromorphic in $r$. The ordinary descendant is likewise
meromorphic. Hence, for an asymptotic contour that avoids their poles,
\begin{equation}
M_\infty\psi_{\mu\nu;m,n}^{L}
=
\psi_{\mu\nu;m,n}^{L},
\qquad
M_\infty\chi_{\mu\nu;m,n}
=
\chi_{\mu\nu;m,n}.
\label{eq:universal_proof_singlevalued}
\end{equation}
Using
\begin{equation}
M_\infty y=y-2\pi i,
\end{equation}
we obtain
\begin{subequations}
\begin{align}
M_\infty\psi_{\mu\nu;m,n}^{\log}
&=
(M_\infty y)
(M_\infty\psi_{\mu\nu;m,n}^{L})
+
M_\infty\chi_{\mu\nu;m,n}
\\[1mm]
&=
(y-2\pi i)\psi_{\mu\nu;m,n}^{L}
+
\chi_{\mu\nu;m,n}
\\[1mm]
&=
\psi_{\mu\nu;m,n}^{\log}
-
2\pi i\,\psi_{\mu\nu;m,n}^{L}.
\end{align}
\end{subequations}
This proves Eqs.~\eqref{eq:universal_left_monodromy} and
\eqref{eq:universal_log_monodromy}.
\end{proof}

For each bi-grade $(m,n)$, introduce the ordered pair
\begin{equation}
\Psi_{\mu\nu;m,n}
=
\begin{pmatrix}
\psi_{\mu\nu;m,n}^{L}
\\[1mm]
\psi_{\mu\nu;m,n}^{\log}
\end{pmatrix}.
\label{eq:descendant_pair_vector}
\end{equation}
Proposition~\ref{prop:universal_asymptotic_monodromy} then implies
\begin{equation}
M_\infty
\Psi_{\mu\nu;m,n}
=
\begin{pmatrix}
1&0
\\[1mm]
-2\pi i&1
\end{pmatrix}
\Psi_{\mu\nu;m,n},
\qquad
m,n\geq0.
\label{eq:universal_monodromy_matrix}
\end{equation}
The same unipotent matrix therefore acts on every logarithmic pair in
the global descendant tower.

Equivalently,
\begin{equation}
\left(
M_\infty-\mathbf 1
\right)
\psi_{\mu\nu;m,n}^{L}
=
0,
\label{eq:MminusI_left}
\end{equation}
whereas
\begin{equation}
\left(
M_\infty-\mathbf 1
\right)
\psi_{\mu\nu;m,n}^{\log}
=
-2\pi i\,\psi_{\mu\nu;m,n}^{L}.
\label{eq:MminusI_log}
\end{equation}
Applying the same operator once more gives
\begin{equation}
\left(
M_\infty-\mathbf 1
\right)^2
\psi_{\mu\nu;m,n}^{\log}
=
0,
\label{eq:MminusI_squared}
\end{equation}
and therefore
\begin{equation}
\left(
M_\infty-\mathbf 1
\right)^2
=
0
\label{eq:global_unipotent_rank_two}
\end{equation}
on the global logarithmic descendant space. The monodromy is thus
unipotent of index two throughout the entire bi-graded tower.

The universal monodromy naturally defines a nilpotent map from each
logarithmic descendant to its ordinary partner. In the next subsection
we construct this operator on the complete global descendant space and
express $M_\infty$ as its unipotent exponential.

\subsection{Nilpotent monodromy operator}
\label{subsec:nilpotent_monodromy_operator}

The universal transformation obtained in
Proposition~\ref{prop:universal_asymptotic_monodromy} defines a
nilpotent operator naturally on the global descendant space. Rather
than introducing this operator from the zero-mode Jordan structure, we
define it directly from the asymptotic bulk monodromy. This makes it
possible to compare, rather than identify by assumption, the analytic
and representation-theoretic nilpotent structures.

Let
\begin{equation}
\mathcal V_{\mathrm{glob}}
=
\operatorname{span}
\left\{
\psi_{\mu\nu;m,n}^{L},
\psi_{\mu\nu;m,n}^{\log}
\ \big|\ 
m,n\in\mathbb Z_{\geq0}
\right\}
\label{eq:global_descendant_space}
\end{equation}
denote the global logarithmic descendant space. Since
\begin{equation}
\left(M_\infty-\mathbf 1\right)^2=0
\label{eq:M_unipotent_recalled}
\end{equation}
on $\mathcal V_{\mathrm{glob}}$, define the monodromy nilpotent by
\begin{equation}
\mathcal N_\infty
:=
-\frac{1}{2\pi i}
\left(
M_\infty-\mathbf 1
\right).
\label{eq:Ninfty_definition}
\end{equation}
The universal monodromy relations immediately imply
\begin{equation}
\mathcal N_\infty
\psi_{\mu\nu;m,n}^{L}
=
0,
\qquad
\mathcal N_\infty
\psi_{\mu\nu;m,n}^{\log}
=
\psi_{\mu\nu;m,n}^{L},
\label{eq:Ninfty_action}
\end{equation}
for every $m,n\geq0$. Consequently,
\begin{equation}
\mathcal N_\infty^2=0
\label{eq:Ninfty_nilpotent}
\end{equation}
throughout $\mathcal V_{\mathrm{glob}}$.

The asymptotic monodromy can therefore be reconstructed from
$\mathcal N_\infty$ as
\begin{equation}
M_\infty
=
\mathbf 1-2\pi i\,\mathcal N_\infty.
\label{eq:M_Ninfty_linear}
\end{equation}
Because $\mathcal N_\infty^2=0$, the exponential truncates after the
linear term
\begin{equation}
M_\infty
=
\exp\!\left(
-2\pi i\,\mathcal N_\infty
\right).
\label{eq:M_Ninfty_exponential}
\end{equation}
Thus the asymptotic continuation defines a rank-two unipotent action on
every nontrivial logarithmic pair in the global descendant space.

For each fixed bi-grade $(m,n)$, the ordered basis
\begin{equation}
\Psi_{\mu\nu;m,n}
=
\begin{pmatrix}
\psi_{\mu\nu;m,n}^{L}
\\[1mm]
\psi_{\mu\nu;m,n}^{\log}
\end{pmatrix}
\end{equation}
gives
\begin{equation}
\mathcal N_\infty
\Psi_{\mu\nu;m,n}
=
\begin{pmatrix}
0&0
\\[1mm]
1&0
\end{pmatrix}
\Psi_{\mu\nu;m,n},
\label{eq:Ninfty_matrix}
\end{equation}
and hence
\begin{equation}
M_\infty
\Psi_{\mu\nu;m,n}
=
\exp\!\left[
-2\pi i
\begin{pmatrix}
0&0
\\[1mm]
1&0
\end{pmatrix}
\right]
\Psi_{\mu\nu;m,n}
=
\begin{pmatrix}
1&0
\\[1mm]
-2\pi i&1
\end{pmatrix}
\Psi_{\mu\nu;m,n}.
\label{eq:M_matrix_from_Ninfty}
\end{equation}
The same nilpotent matrix is therefore realized at every bi-grade.

We now compare $\mathcal N_\infty$ with the nilpotent map
associated with the established Jordan action of the zero modes. Let
$N$ denote the representation-theoretic nilpotent defined on the
primary logarithmic pair by
\begin{equation}
N\psi_{\mu\nu}^{L}=0,
\qquad
N\psi_{\mu\nu}^{\log}
=
\psi_{\mu\nu}^{L}.
\label{eq:N_primary_comparison}
\end{equation}
The standard Virasoro commutation relations imply that its extension
through the global descendant tower acts as
\begin{equation}
N\psi_{\mu\nu;m,n}^{L}=0,
\qquad
N\psi_{\mu\nu;m,n}^{\log}
=
\psi_{\mu\nu;m,n}^{L}.
\label{eq:N_descendant_comparison}
\end{equation}
Comparing Eqs.~\eqref{eq:Ninfty_action} and
\eqref{eq:N_descendant_comparison}, we obtain
\begin{equation}
\mathcal N_\infty
=
N
\qquad
\text{on }\mathcal V_{\mathrm{glob}}.
\label{eq:Ninfty_equals_N}
\end{equation}
This completes the monodromy analysis of the global descendant space.
The next step is to determine whether the bulk-derived nilpotent map is
compatible with the complete global conformal action, including the
raising generators $L_1$ and $\bar L_1$. We address this question in
the following section.

\section{Full Global Conformal Action and Monodromy Intertwining}
\label{sec:global_action}

The preceding sections constructed the asymptotic monodromy directly
from the complex radial structure of the bulk logarithmic
graviton and extended its action to the complete global descendant
space. This construction produced the nilpotent operator
$\mathcal N_\infty$ independently of the global conformal
representation and subsequently showed that its action agrees with the
nilpotent map associated with the logarithmic pair. We now ask a
stronger question: whether this bulk-derived monodromy is compatible
with the full
$SL(2,\mathbb R)_L\times SL(2,\mathbb R)_R$ action on the descendant
space. This requires more than the lowering relations used to construct
the descendants. We must examine the action of the zero modes and the
raising generators and determine whether the monodromy nilpotent
intertwines the complete global conformal representation. We begin with
$L_0$ and $\bar L_0$, whose action makes explicit how the logarithmic
mixing is distributed over the bi-graded descendant tower.

\subsection{Action of $L_0$ and $\bar L_0$}
\label{subsec:L0_Lbar0_action}

The ordinary left-moving primary has conformal weights $(h,\bar h)=(2,0)$, 
so that
\begin{equation}
L_0\psi_{\mu\nu}^{L}
=
2\psi_{\mu\nu}^{L},
\qquad
\bar L_0\psi_{\mu\nu}^{L}
=
0.
\label{eq:L0_primary_left_section7}
\end{equation}
For the logarithmic partner, the tensorial Leibniz rule therefore gives
\begin{equation}
L_0\psi_{\mu\nu}^{\log}
=
L_0\left(y\psi_{\mu\nu}^{L}\right)
=
yL_0\psi_{\mu\nu}^{L}
+
(L_0y)\psi_{\mu\nu}^{L}
=
2\psi_{\mu\nu}^{\log}
+
\frac{1}{2}\psi_{\mu\nu}^{L},
\label{eq:L0_log_primary_section7}
\end{equation}
and
\begin{equation}
\bar L_0\psi_{\mu\nu}^{\log}
=
\bar L_0\left(y\psi_{\mu\nu}^{L}\right)
=
y\bar L_0\psi_{\mu\nu}^{L}
+
(\bar L_0y)\psi_{\mu\nu}^{L}
=
\frac{1}{2}\psi_{\mu\nu}^{L}.
\label{eq:Lbar0_log_primary_section7}
\end{equation}
Their sum,
\begin{equation}
D:=L_0+\bar L_0,
\label{eq:D_definition_section7}
\end{equation}
acts as
\begin{equation}
D\psi_{\mu\nu}^{L}
=
2\psi_{\mu\nu}^{L},
\qquad
D\psi_{\mu\nu}^{\log}
=
2\psi_{\mu\nu}^{\log}
+
\psi_{\mu\nu}^{L}.
\label{eq:D_primary_section7}
\end{equation}
Hence the nilpotent part of the total zero-mode action has the
canonical normalization
\begin{equation}
(D-2)\psi_{\mu\nu}^{\log}
=
\psi_{\mu\nu}^{L},
\qquad
(D-2)\psi_{\mu\nu}^{L}
=
0.
\label{eq:D_nilpotent_primary_section7}
\end{equation}
This is precisely the normalization with which the monodromy
nilpotent $\mathcal N_\infty$ was obtained in
Sec.~\ref{subsec:nilpotent_monodromy_operator}.

We now extend the zero-mode action to arbitrary bi-grade $(m,n)$.
Using the global conformal commutators
\begin{equation}
[L_0,L_{-1}]=L_{-1},
\qquad
[\bar L_0,\bar L_{-1}]=\bar L_{-1},
\qquad
[L_0,\bar L_{-1}]
=
[\bar L_0,L_{-1}]
=
0,
\label{eq:zero_lowering_commutators}
\end{equation}
the ordinary descendants satisfy
\begin{equation}
L_0\psi_{\mu\nu;m,n}^{L}
=
(2+m)\psi_{\mu\nu;m,n}^{L},
\label{eq:L0_left_descendant}
\end{equation}
and
\begin{equation}
\bar L_0\psi_{\mu\nu;m,n}^{L}
=
n\,\psi_{\mu\nu;m,n}^{L}.
\label{eq:Lbar0_left_descendant}
\end{equation}
Thus the ordinary descendant at bi-grade $(m,n)$ carries weights
\begin{equation}
(h_{m,n},\bar h_{m,n})
=
(2+m,n).
\label{eq:descendant_weights}
\end{equation}

For the logarithmic descendants, the same commutators may be applied
to their definition
\begin{equation}
\psi_{\mu\nu;m,n}^{\log}
=
L_{-1}^{\,m}\bar L_{-1}^{\,n}
\psi_{\mu\nu}^{\log}.
\label{eq:log_descendant_definition_section7}
\end{equation}
Using the primary relations
\eqref{eq:L0_log_primary_section7} and
\eqref{eq:Lbar0_log_primary_section7}, one obtains
\begin{equation}
L_0\psi_{\mu\nu;m,n}^{\log}
=
(2+m)\psi_{\mu\nu;m,n}^{\log}
+
\frac{1}{2}\psi_{\mu\nu;m,n}^{L},
\label{eq:L0_log_descendant}
\end{equation}
and
\begin{equation}
\bar L_0\psi_{\mu\nu;m,n}^{\log}
=
n\,\psi_{\mu\nu;m,n}^{\log}
+
\frac{1}{2}\psi_{\mu\nu;m,n}^{L}.
\label{eq:Lbar0_log_descendant}
\end{equation}
Adding the two relations gives
\begin{equation}
D\psi_{\mu\nu;m,n}^{\log}
=
(2+m+n)\psi_{\mu\nu;m,n}^{\log}
+
\psi_{\mu\nu;m,n}^{L},
\label{eq:D_log_descendant}
\end{equation}
while
\begin{equation}
D\psi_{\mu\nu;m,n}^{L}
=
(2+m+n)\psi_{\mu\nu;m,n}^{L}.
\label{eq:D_left_descendant}
\end{equation}

Consequently, at every bi-grade the non-semisimple part of the total
zero-mode action is
\begin{equation}
\left[
D-(2+m+n)\mathbf 1
\right]
\psi_{\mu\nu;m,n}^{\log}
=
\psi_{\mu\nu;m,n}^{L},
\label{eq:D_nilpotent_descendant_log}
\end{equation}
and
\begin{equation}
\left[
D-(2+m+n)\mathbf 1
\right]
\psi_{\mu\nu;m,n}^{L}
=
0.
\label{eq:D_nilpotent_descendant_left}
\end{equation}
Comparison with
Eq.~\eqref{eq:Ninfty_action} therefore yields
\begin{equation}
\left.
\left[
D-(2+m+n)\mathbf 1
\right]
\right|_{\mathcal V_{m,n}}
=
\left.
\mathcal N_\infty
\right|_{\mathcal V_{m,n}},
\label{eq:D_nilpotent_equals_Ninfty}
\end{equation}
where
\begin{equation}
\mathcal V_{m,n}
=
\operatorname{span}
\left\{
\psi_{\mu\nu;m,n}^{L},
\psi_{\mu\nu;m,n}^{\log}
\right\}.
\label{eq:Vmn_definition}
\end{equation}
Thus the nilpotent operator extracted independently from the bulk
monodromy agrees, at every descendant level, with the
non-diagonalizable part of the total zero-mode action.

In the ordered basis
\begin{equation}
\Psi_{\mu\nu;m,n}
=
\begin{pmatrix}
\psi_{\mu\nu;m,n}^{L}
\\[1mm]
\psi_{\mu\nu;m,n}^{\log}
\end{pmatrix},
\end{equation}
the zero modes are represented by
\begin{equation}
L_0
=
\begin{pmatrix}
2+m & 0
\\[1mm]
\frac{1}{2} & 2+m
\end{pmatrix},
\qquad
\bar L_0
=
\begin{pmatrix}
n & 0
\\[1mm]
\frac{1}{2} & n
\end{pmatrix},
\label{eq:zero_mode_matrices}
\end{equation}
and therefore
\begin{equation}
D
=
\begin{pmatrix}
2+m+n & 0
\\[1mm]
1 & 2+m+n
\end{pmatrix}.
\label{eq:D_matrix_descendant}
\end{equation}
The monodromy nilpotent is correspondingly
\begin{equation}
\mathcal N_\infty
=
\begin{pmatrix}
0&0
\\[1mm]
1&0
\end{pmatrix},
\label{eq:Ninfty_matrix_section7}
\end{equation}
and one verifies directly that
\begin{equation}
[\mathcal N_\infty,L_0]
=
[\mathcal N_\infty,\bar L_0]
=
0
\label{eq:Ninfty_zero_mode_commutators}
\end{equation}
on every $\mathcal V_{m,n}$. Since
\begin{equation}
M_\infty
=
e^{-2\pi i\mathcal N_\infty},
\end{equation}
it follows that
\begin{equation}
[M_\infty,L_0]
=
[M_\infty,\bar L_0]
=
0.
\label{eq:monodromy_zero_mode_intertwining}
\end{equation}

The asymptotic bulk monodromy therefore intertwines the action of both
global zero modes, providing the first part of the full
intertwining problem. Compatibility with
$SL(2,\mathbb R)_L\times SL(2,\mathbb R)_R$, however, also requires
the bulk-derived nilpotent map to intertwine the raising operators.
We therefore turn next to the action of $L_1$.

\subsection{Action of $L_1$}
\label{subsec:L1_action}

We next examine the action of the left-moving raising generator
$L_1$. Unlike the zero modes, $L_1$ does not preserve the bi-grade
$(m,n)$, but maps a descendant at level $(m,n)$ to level
$(m-1,n)$. The logarithmic sector requires particular care because the
non-diagonalizable action of $L_0$ generates an additional ordinary
descendant contribution. This provides a nontrivial test of whether the
bulk-derived monodromy nilpotent $\mathcal N_\infty$ intertwines the
global conformal action.

We begin at the primary level. The ordinary left-moving mode is a
highest-weight state and therefore satisfies
\begin{equation}
L_1\psi_{\mu\nu}^{L}=0.
\label{eq:L1_left_primary}
\end{equation}
Under $L_1$, the logarithmic partner transforms as
\begin{equation}
L_1\psi_{\mu\nu}^{\log}
=
L_1
\left(
y\,\psi_{\mu\nu}^{L}
\right)
=
(L_1y)\psi_{\mu\nu}^{L}
+
yL_1\psi_{\mu\nu}^{L}
=
0,
\end{equation}
and therefore
\begin{equation}
L_1\psi_{\mu\nu}^{L}
=
L_1\psi_{\mu\nu}^{\log}
=
0,
\label{eq:L1_primary_pair}
\end{equation}
with both members of the logarithmic primary pair annihilated by
the left-moving global raising generator.

We now consider an arbitrary descendant. The global conformal algebra
contains
\begin{equation}
[L_1,L_{-1}]
=
2L_0,
\qquad
[L_0,L_{-1}]
=
L_{-1}.
\label{eq:sl2_commutators_L1}
\end{equation}
Repeated use of these relations gives
\begin{equation}
[L_1,L_{-1}^{\,m}]
=
mL_{-1}^{\,m-1}
\left(
2L_0+m-1
\right).
\label{eq:L1_Lminusm_commutator}
\end{equation}
Since
\begin{equation}
[L_1,\bar L_{-1}]=0,
\label{eq:L1_barLminus_commutator}
\end{equation}
the right-moving descendant label $n$ is unaffected by the action of
$L_1$.

For the ordinary descendant,
\begin{equation}
L_1\psi_{\mu\nu;m,n}^{L}
=
L_1
L_{-1}^{\,m}\bar L_{-1}^{\,n}
\psi_{\mu\nu}^{L}
=
[L_1,L_{-1}^{\,m}]
\bar L_{-1}^{\,n}
\psi_{\mu\nu}^{L},
\end{equation}
where Eq.~\eqref{eq:L1_left_primary} has been used. Since
\begin{equation}
L_0\psi_{\mu\nu}^{L}
=
2\psi_{\mu\nu}^{L},
\end{equation}
Eq.~\eqref{eq:L1_Lminusm_commutator} yields
\begin{equation}
L_1\psi_{\mu\nu;m,n}^{L}
=
m(m+3)\,
\psi_{\mu\nu;m-1,n}^{L},
\qquad
m\geq1.
\label{eq:L1_left_descendant_action}
\end{equation}
For $m=0$, the right-hand side is understood to vanish, consistently
with the highest-weight condition.

The logarithmic descendant contains an additional contribution because
$L_0$ acts non-diagonalizably on the logarithmic primary. Using Eq. (\ref{eq:L0_log_primary_section7}), we find
\begin{subequations}
\begin{align}
L_1\psi_{\mu\nu;m,n}^{\log}
&=
[L_1,L_{-1}^{\,m}]
\bar L_{-1}^{\,n}
\psi_{\mu\nu}^{\log}
\\[1mm]
&=
mL_{-1}^{\,m-1}\bar L_{-1}^{\,n}
\left(
2L_0+m-1
\right)
\psi_{\mu\nu}^{\log}
\\[1mm]
&=
mL_{-1}^{\,m-1}\bar L_{-1}^{\,n}
\left[
(m+3)\psi_{\mu\nu}^{\log}
+
\psi_{\mu\nu}^{L}
\right].
\end{align}
\end{subequations}
Therefore
\begin{equation}
L_1\psi_{\mu\nu;m,n}^{\log}
=
m(m+3)\,
\psi_{\mu\nu;m-1,n}^{\log}
+
m\,\psi_{\mu\nu;m-1,n}^{L},
\qquad
m\geq1.
\label{eq:L1_log_descendant_action}
\end{equation}

The second term in
Eq.~\eqref{eq:L1_log_descendant_action} is a direct consequence of the
logarithmic mixing in the $L_0$ action. It has no analogue in the
ordinary descendant tower. Thus, although $L_1$ preserves the
two-component logarithmic descendant space globally, its action is not
simply two identical copies of the ordinary highest-weight
representation.

For fixed $(m,n)$ with $m\geq1$, the action may be displayed as
\begin{equation}
L_1
\begin{pmatrix}
\psi_{\mu\nu;m,n}^{L}
\\[1mm]
\psi_{\mu\nu;m,n}^{\log}
\end{pmatrix}
=
\begin{pmatrix}
m(m+3) & 0
\\[1mm]
m & m(m+3)
\end{pmatrix}
\begin{pmatrix}
\psi_{\mu\nu;m-1,n}^{L}
\\[1mm]
\psi_{\mu\nu;m-1,n}^{\log}
\end{pmatrix}.
\label{eq:L1_pair_matrix}
\end{equation}
Equivalently, introducing the nilpotent matrix
\begin{equation}
\mathcal N_\infty
=
\begin{pmatrix}
0&0
\\[1mm]
1&0
\end{pmatrix},
\end{equation}
the level-changing action takes the form
\begin{equation}
L_1
=
m(m+3)\,\mathbf 1
+
m\,\mathcal N_\infty
\label{eq:L1_nilpotent_form}
\end{equation}
as a map from $\mathcal V_{m,n}$ to
$\mathcal V_{m-1,n}$.

We can now test the intertwining property directly. Acting first with
$\mathcal N_\infty$ and then with $L_1$ on a logarithmic descendant
gives
\begin{equation}
L_1\mathcal N_\infty
\psi_{\mu\nu;m,n}^{\log}
=
L_1\psi_{\mu\nu;m,n}^{L}
=
m(m+3)
\psi_{\mu\nu;m-1,n}^{L}.
\label{eq:L1_N_log}
\end{equation}
Reversing the order gives
\begin{equation}
\mathcal N_\infty
L_1\psi_{\mu\nu;m,n}^{\log}
=
\mathcal N_\infty
\left[
m(m+3)\psi_{\mu\nu;m-1,n}^{\log}
+
m\psi_{\mu\nu;m-1,n}^{L}
\right]
=
m(m+3)
\psi_{\mu\nu;m-1,n}^{L},
\label{eq:N_L1_log}
\end{equation}
because
\begin{equation}
\mathcal N_\infty
\psi_{\mu\nu;m-1,n}^{L}
=
0.
\end{equation}
The same relation is immediate on an ordinary descendant. Hence
\begin{equation}
[\mathcal N_\infty,L_1]=0
\label{eq:Ninfty_L1_commutator}
\end{equation}
on the complete global descendant space.

Since
\begin{equation}
M_\infty
=
\exp(-2\pi i\mathcal N_\infty),
\end{equation}
the same result implies
\begin{equation}
[M_\infty,L_1]=0.
\label{eq:M_L1_commutator}
\end{equation}
Thus analytic continuation along the asymptotic radial contour
intertwines the action of the left-moving raising generator.

This result is nontrivial in the logarithmic sector because
$L_1$ itself contains the additional mixing term displayed in
Eq.~\eqref{eq:L1_log_descendant_action}. The intertwining relation does
not follow from treating the logarithmic and ordinary descendant towers
as identical representations. Rather, the additional term lies in the
image of the nilpotent map and is annihilated by a subsequent action of
$\mathcal N_\infty$, precisely as required for the two operations to
commute.

Together with
Eq.~\eqref{eq:monodromy_zero_mode_intertwining}, we have therefore
established
\begin{equation}
[\mathcal N_\infty,L_0]
=
[\mathcal N_\infty,L_1]
=
0,
\label{eq:left_partial_intertwining}
\end{equation}
and equivalently
\begin{equation}
[M_\infty,L_0]
=
[M_\infty,L_1]
=
0.
\label{eq:left_partial_M_intertwining}
\end{equation}
The corresponding test for the right-moving raising generator is
carried out in the next subsection.

\subsection{Action of $\bar L_1$}
\label{subsec:Lbar1_action}

We now turn to the right-moving raising generator
$\bar L_1$. The derivation parallels the left-moving analysis of
Sec.~\ref{subsec:L1_action}, but the resulting coefficients are
different because the ordinary left-moving primary has right-moving
weight $\bar h=0$. This distinction makes the action of $\bar L_1$ particularly instructive in the logarithmic sector.

At the primary level, the ordinary left-moving mode and the logarithmic partner both obey
\begin{equation}
\bar L_1\psi_{\mu\nu}^{L}
=
\bar L_1\psi_{\mu\nu}^{\log}
=
0.
\label{eq:Lbar1_left_primary}
\end{equation}

For arbitrary right-moving descendant level, the global conformal
algebra gives
\begin{equation}
[\bar L_1,\bar L_{-1}]
=
2\bar L_0,
\qquad
[\bar L_0,\bar L_{-1}]
=
\bar L_{-1},
\label{eq:bar_sl2_commutators}
\end{equation}
and hence
\begin{equation}
[\bar L_1,\bar L_{-1}^{\,n}]
=
n\bar L_{-1}^{\,n-1}
\left(
2\bar L_0+n-1
\right).
\label{eq:Lbar1_Lbarminusn_commutator}
\end{equation}
Since
\begin{equation}
[\bar L_1,L_{-1}]=0,
\label{eq:Lbar1_Lminus_commutator}
\end{equation}
the left-moving descendant index $m$ is unaffected by the action of
$\bar L_1$.

For the ordinary descendant,
\begin{equation}
\bar L_1\psi_{\mu\nu;m,n}^{L}
=
\bar L_1
L_{-1}^{\,m}\bar L_{-1}^{\,n}
\psi_{\mu\nu}^{L}
=
L_{-1}^{\,m}
[\bar L_1,\bar L_{-1}^{\,n}]
\psi_{\mu\nu}^{L},
\end{equation}
where the primary condition
\eqref{eq:Lbar1_left_primary} has been used. Making use of Eq. (\ref{eq:L0_primary_left_section7}), Eq.~\eqref{eq:Lbar1_Lbarminusn_commutator} yields
\begin{equation}
\bar L_1\psi_{\mu\nu;m,n}^{L}
=
n(n-1)\,
\psi_{\mu\nu;m,n-1}^{L},
\qquad
n\geq1.
\label{eq:Lbar1_left_descendant_action}
\end{equation}
For $n=0$, the right-hand side vanishes by the highest-weight
condition.

The logarithmic descendant again acquires an additional ordinary
contribution. From Eq. (\ref{eq:Lbar0_log_primary_section7}), one finds
\begin{subequations}
\begin{align}
\bar L_1\psi_{\mu\nu;m,n}^{\log}
&=
L_{-1}^{\,m}
[\bar L_1,\bar L_{-1}^{\,n}]
\psi_{\mu\nu}^{\log}
\\[1mm]
&=
nL_{-1}^{\,m}\bar L_{-1}^{\,n-1}
\left(
2\bar L_0+n-1
\right)
\psi_{\mu\nu}^{\log}
\\[1mm]
&=
nL_{-1}^{\,m}\bar L_{-1}^{\,n-1}
\left[
(n-1)\psi_{\mu\nu}^{\log}
+
\psi_{\mu\nu}^{L}
\right].
\end{align}
\end{subequations}
Therefore
\begin{equation}
\bar L_1\psi_{\mu\nu;m,n}^{\log}
=
n(n-1)\,
\psi_{\mu\nu;m,n-1}^{\log}
+
n\,\psi_{\mu\nu;m,n-1}^{L},
\qquad
n\geq1.
\label{eq:Lbar1_log_descendant_action}
\end{equation}

The right-moving action differs qualitatively from its left-moving
counterpart because the semisimple coefficient is controlled by
$\bar h=0$. In particular, at the first right-moving level,
$n=1$, one obtains
\begin{equation}
\bar L_1\psi_{\mu\nu;m,1}^{L}
=
0,
\label{eq:Lbar1_n1_left}
\end{equation}
whereas
\begin{equation}
\bar L_1\psi_{\mu\nu;m,1}^{\log}
=
\psi_{\mu\nu;m,0}^{L}.
\label{eq:Lbar1_n1_log}
\end{equation}
Thus the first logarithmic right-moving descendant is mapped directly
onto the corresponding ordinary state. This provides a particularly
transparent manifestation of the indecomposable structure.

For fixed $(m,n)$ with $n\geq1$, the action may be written as
\begin{equation}
\bar L_1
\begin{pmatrix}
\psi_{\mu\nu;m,n}^{L}
\\[1mm]
\psi_{\mu\nu;m,n}^{\log}
\end{pmatrix}
=
\begin{pmatrix}
n(n-1) & 0
\\[1mm]
n & n(n-1)
\end{pmatrix}
\begin{pmatrix}
\psi_{\mu\nu;m,n-1}^{L}
\\[1mm]
\psi_{\mu\nu;m,n-1}^{\log}
\end{pmatrix}.
\label{eq:Lbar1_pair_matrix}
\end{equation}
Equivalently,
\begin{equation}
\bar L_1
=
n(n-1)\,\mathbf 1
+
n\,\mathcal N_\infty
\label{eq:Lbar1_nilpotent_form}
\end{equation}
as a map from $\mathcal V_{m,n}$ to
$\mathcal V_{m,n-1}$.

We may now test the intertwining relation. Acting first with
$\mathcal N_\infty$ gives
\begin{equation}
\bar L_1
\mathcal N_\infty
\psi_{\mu\nu;m,n}^{\log}
=
\bar L_1
\psi_{\mu\nu;m,n}^{L}
=
n(n-1)
\psi_{\mu\nu;m,n-1}^{L}.
\label{eq:Lbar1_N_log}
\end{equation}
Reversing the order yields
\begin{equation}
\mathcal N_\infty
\bar L_1
\psi_{\mu\nu;m,n}^{\log}
=
\mathcal N_\infty
\left[
n(n-1)\psi_{\mu\nu;m,n-1}^{\log}
+
n\psi_{\mu\nu;m,n-1}^{L}
\right]
=
n(n-1)
\psi_{\mu\nu;m,n-1}^{L},
\label{eq:N_Lbar1_log}
\end{equation}
since
\begin{equation}
\mathcal N_\infty
\psi_{\mu\nu;m,n-1}^{L}
=
0.
\end{equation}
The same identity is immediate on an ordinary descendant. Hence
\begin{equation}
[\mathcal N_\infty,\bar L_1]=0
\label{eq:Ninfty_Lbar1_commutator}
\end{equation}
throughout the global descendant space.

Since
\begin{equation}
M_\infty
=
\exp(-2\pi i\mathcal N_\infty),
\end{equation}
we also obtain
\begin{equation}
[M_\infty,\bar L_1]=0.
\label{eq:M_Lbar1_commutator}
\end{equation}
Thus the asymptotic bulk monodromy intertwines the right-moving raising
action as well.

Combining the results of this subsection with
Secs.~\ref{subsec:L0_Lbar0_action} and \ref{subsec:L1_action}, we have
now established
\begin{equation}
[\mathcal N_\infty,L_0]
=
[\mathcal N_\infty,\bar L_0]
=
[\mathcal N_\infty,L_1]
=
[\mathcal N_\infty,\bar L_1]
=
0.
\label{eq:nilpotent_intertwining_so_far}
\end{equation}
Together with the defining descendant relations generated by
$L_{-1}$ and $\bar L_{-1}$, this provides the ingredients required to
show that $\mathcal N_\infty$ intertwines the complete global
$SL(2,\mathbb R)_L\times SL(2,\mathbb R)_R$ representation. We make
this statement precise in the following subsections.

\subsection{The nilpotent intertwiner and compatibility with $SL(2,\mathbb R)_L\times SL(2,\mathbb R)_R$}
\label{subsec:nilpotent_intertwiner}

The calculations of the preceding subsections show that the
bulk-derived nilpotent operator $\mathcal N_\infty$ is compatible with
the zero-mode and raising-generator actions. The lowering generators
were already built into the descendant construction itself. We now
collect these facts into a single representation-theoretic statement
to show that $\mathcal N_\infty$ intertwines the complete global
conformal algebra.

\begin{proposition}[Nilpotent intertwiner of the global conformal module]
\label{prop:nilpotent_intertwiner}
On the global logarithmic descendant space
\begin{equation}
\mathcal V_{\mathrm{glob}}
=
\operatorname{span}
\left\{
\psi_{\mu\nu;m,n}^{L},
\psi_{\mu\nu;m,n}^{\log}
\ \big|\ 
m,n\in\mathbb Z_{\geq0}
\right\},
\end{equation}
the bulk monodromy nilpotent $\mathcal N_\infty$ satisfies
\begin{equation}
[\mathcal N_\infty,L_a]
=
[\mathcal N_\infty,\bar L_a]
=
0,
\qquad
a=-1,0,1.
\label{eq:full_N_intertwining}
\end{equation}
Hence $\mathcal N_\infty$ is an endomorphism of the global
$SL(2,\mathbb R)_L\times SL(2,\mathbb R)_R$ representation.
\end{proposition}

\begin{proof}
For $a=-1$, the result follows directly from the descendant definitions. For $a=0$, it follows
from the zero-mode actions
\eqref{eq:L0_log_descendant} and
\eqref{eq:Lbar0_log_descendant}. For $a=1$, it follows from the
raising-generator relations
\eqref{eq:L1_log_descendant_action} and
\eqref{eq:Lbar1_log_descendant_action}. The commutators vanish on both
the ordinary and logarithmic descendants, and therefore on all of
$\mathcal V_{\mathrm{glob}}$.
\end{proof}

Having established that $\mathcal N_\infty$ intertwines all six global
generators, the compatibility between asymptotic bulk monodromy and the complete $SL(2,\mathbb R)_L\times SL(2,\mathbb R)_R$ action follows immediately. Indeed, since it is established that $M_\infty=\mathbf 1-2\pi i\mathcal N_\infty$, Proposition~\ref{prop:nilpotent_intertwiner} immediately implies
\begin{equation}
[M_\infty,L_a]
=
[M_\infty,\bar L_a]
=
0,
\qquad
a=-1,0,1.
\label{eq:M_full_global_commutant}
\end{equation}
Hence $M_\infty$ belongs to the commutant of the
$SL(2,\mathbb R)_L\times SL(2,\mathbb R)_R$ representation carried by
the global logarithmic descendant space.

\section{Summary and Discussion}
\label{sec:conclusions}

In this work, we have investigated the complex radial analytic
structure of the logarithmic graviton in critical topologically massive
gravity and its realization across the global conformal descendant
space. The starting point was the Grumiller--Johansson
logarithmic mode rather than an abstract Jordan representation. This
allowed the monodromy to be derived directly from the
coordinate-dependent bulk solution and subsequently compared with the
known indecomposable structure of the logarithmic sector.

The principal results of the paper can therefore be summarized as
follows:
\begin{enumerate}[(i)]

\item
The complex radial continuation of the logarithmic graviton is
classified by the winding numbers about $r=0,\pm i$, leading to the
asymptotic shift $\Delta_\infty y=-2\pi i$.

\item
The global logarithmic descendants satisfy
\begin{equation}
\psi_{\mu\nu;m,n}^{\log}
=
y\,\psi_{\mu\nu;m,n}^{L}
+
\chi_{\mu\nu;m,n},
\end{equation}
with a log-free meromorphic remainder
$\chi_{\mu\nu;m,n}$.

\item
The complete global descendant tower consequently carries the
universal bulk monodromy
\begin{equation}
M_\infty\psi_{\mu\nu;m,n}^{\log}
=
\psi_{\mu\nu;m,n}^{\log}
-
2\pi i\,\psi_{\mu\nu;m,n}^{L}.
\end{equation}

\item
The nilpotent operator extracted from this monodromy intertwines the
complete global conformal action,
\begin{equation}
[\mathcal N_\infty,L_a]
=
[\mathcal N_\infty,\bar L_a]
=
0,
\qquad
a=-1,0,1.
\end{equation}
Equivalently,
\begin{equation}
[M_\infty,L_a]
=
[M_\infty,\bar L_a]
=
0,
\qquad
a=-1,0,1.
\end{equation}
and the asymptotic bulk monodromy therefore defines an intertwining
endomorphism of the global
$SL(2,\mathbb R)_L\times SL(2,\mathbb R)_R$ representation.
\end{enumerate}
These results shift the role of monodromy from an algebraic
reinterpretation of an already known Jordan cell to an independently
derived property of the bulk fields. The known representation-theoretic structure then emerges as the structure with which this bulk monodromy is compatible.

The method developed here is not tied conceptually to the rank-two degeneration of critical TMG. Massive gravity theories with higher-rank critical points provide natural settings in which the linearized solution space already contains rank-three or higher
logarithmic structures at the single-particle level. Applying the same
sequence---bulk solution, complex radial continuation, descendant
decomposition, and global intertwining---could clarify how the order of
the bulk logarithmic degeneration is reflected in the associated
monodromy representation. This extension lies beyond the scope of the present paper and is left for future work.

\paragraph{Acknowledgements}
The author would like to thank Daniel Grumiller for his response to the author's query concerning the linearly growing contribution indexed by $\rho$ in the Fefferman--Graham expansion of the metric \cite{Grumiller:2008qz,Grumiller:2013at}. The author also thanks anonymous referees for helping improve the clarity, scope, and presentation of this work. The author acknowledges financial support from the Department of Physics at the University of Pretoria.

\clearpage

\bibliographystyle{utphys}
\bibliography{sample}

\end{document}